\documentclass[journal,comsoc,twoside]{IEEEtran}
\IEEEoverridecommandlockouts
\usepackage{cite}
\usepackage{amsmath,amssymb,amsfonts}
\usepackage{graphicx}
\usepackage{textcomp}
\usepackage{amsthm}
\usepackage{subfig}
\usepackage{algorithm}
\usepackage{algpseudocode}
\usepackage{mathtools}

\usepackage{tcolorbox}
\newtheorem{theorem}{Theorem}

\DeclareMathOperator*{\argmin}{arg\,min}
\usepackage{hyperref}
\usepackage{xcolor}
\def\BibTeX{{\rm B\kern-.05em{\sc i\kern-.025em b}\kern-.08em
    T\kern-.1667em\lower.7ex\hbox{E}\kern-.125emX}}

\begin{document}

\title{Goal-Oriented Status Updating for Real-time Remote Inference over Networks with Two-Way~Delay}

\author{\IEEEauthorblockN{Çağrı Arı, Md Kamran Chowdhury Shisher, Yin Sun, and Elif Uysal}
\thanks{An earlier version of this paper was presented in part at the IEEE International Symposium on Information Theory (ISIT) 2024~\cite{ari2023goal}. \par
Çağrı Arı and Elif Uysal are with the Department of
Electrical and Electronics Engineering, Middle East Technical University (METU), 06800 Ankara, Turkiye (e-mail: ari.cagri@metu.edu.tr; uelif@metu.edu.tr). \par
Md Kamran Chowdhury Shisher is with the Elmore Family School of Electrical and Computer Engineering, Purdue University, West Lafayette, IN, 47907, USA (e-mail: mshisher@purdue.edu). \par
Yin Sun is with the Department of Electrical and Computer Engineering, Auburn University, Auburn, AL, 36849, USA (e-mail: yzs0078@auburn.edu). \par
This work was supported in part by grant no. 101122990-GO SPACE-ERC-2023-AdG and NSF grant CNS-2239677. \par
Çağrı Arı was also supported by Turk Telekom within the framework of the 5G and Beyond Joint Graduate Support Programme, coordinated by the Information and Communication Technologies Authority.}}

\maketitle

\begin{abstract}
We study a setting where an intelligent model (e.g., a pre-trained neural network) infers the real-time value of a target signal using data samples transmitted from a remote source. The transmission scheduler decides (i) the freshness of packets, (ii) their length (i.e., the number of samples they contain), and (iii) when they should be transmitted. The freshness is quantified using the Age of Information (AoI), and the inference quality for a given packet length is a general function of AoI. Previous works assumed i.i.d. transmission delays with immediate feedback or were restricted to the case where inference performance degrades as the input data ages. Our formulation, in addition to capturing non-monotone age dependence, also covers Markovian delay on both forward and feedback links. We model this as an infinite-horizon average-cost Semi-Markov Decision Process. We obtain a closed-form solution that decides on (i) and (iii) for any constant packet length. The solution for when to transmit is an index-based threshold policy, where the index function is expressed in terms of the delay state and AoI at the receiver. In contrast, the freshness of the selected packet is a function of only the delay state. We then separately optimize the value of the constant packet length. Moreover, we also develop an index-based threshold policy for the time-variable packet length case, which allows a complexity reduction. In simulation results, we observe that our goal-oriented scheduler drops inference error down to one-sixth with respect to the age-based scheduling of unit-length packets.
\end{abstract}

\section{Introduction}
Today, many previously complex computational problems enjoy the aid of AI. Among these AI-aided approaches are intelligent models, e.g., neural networks. These models have played a key role in enabling certain technologies as digital twins, industrial robotics, and self-driving cars. However, these models rely on data to be delivered to them, which, in the remote computation context, depends on the network delivering the relevant data in a timely fashion. This is especially interesting in the case of unavoidable and time-varying transmission delays and is within the realm of Goal-oriented Networking for AI \cite{uysal2022semantic, uysal2024goal}.

An extreme example of monitoring over long-distance network connections is a spacecraft digital twin. Consider training a digital twin on Earth based on data coming from rovers on Mars or calibrating the virtual replica of the spacecraft/satellite using live telemetry data~\cite{digitaltwinsatelliteCommercial, digitaltwinassistedStorage}. Furthermore, in recent years, digital twins have also gained importance in the monitoring of processes in industrial environments, especially in hazardous or hard-to-reach environments or when it is impractical to be on site for monitoring.\footnote{A screenshot of the digital twin of the METU high voltage testing laboratory is available at \url{https://cng-eee.metu.edu.tr/wp-content/Digital_twin.pdf?_t=1726991660}. The digital twin can be monitored over the Internet, and device settings in the physical experiment configuration on site can be remotely changed on demand directly on the twin.} In this paper, we focus on \textit{goal-oriented} communication design for real-time estimation and inference performed by intelligent models~\cite{shisher2022does, shisher2024timely, shisher2023learning}, which are often essential for such applications. 

Goal-oriented communication has been recently introduced to address the scalability problem in next-generation massive machine-type communication networks. This concept emphasizes that networks can facilitate the effective accomplishment of tasks at the destination rather than merely addressing the transmission problem~\cite{shannon1949mathematical}, which aims to reliably transmit data produced by a source. Solving the effective communication problem with efficient use of network resources requires combining the data generation and transmission processes, ensuring that the samples most relevant to the computation at the destination are delivered in a timely manner~\cite{uysal2022semantic}. This implies that the communication link should select which samples to transmit based on the state of the network (e.g., the delay state). Recent efforts in the communication and control communities highlight the value of this approach in reducing communication requirements while maintaining application performance~\cite{voortman2023remote, sagduyu2023multi, merluzzi2022effective, talli2023semantic, peng2024online, li2025freshness}.

Age of Information (AoI) at time $t$, denoted by $\Delta(t)$, is defined in \cite{kaul2012real} as $\Delta(t) = t - U_{t}$, where $U_{t}$ is the generation time of the most recently delivered data packet. The application-layer real-time estimation and inference performance is mapped as a function of the AoI. This function is all that is required for the link layer to operate in a goal-oriented manner. This work builds upon the findings presented in a series of papers \cite{shisher2022does, sun2019sampling, shisher2024timely, shisher2023learning, sun2017update, sun2019sampling2, shisher2024monotonicity}, which first demonstrated the effectiveness of AoI function as a surrogate metric for enhancing timely communication for remote estimation and inference.

In this paper, we study a \textit{remote inference} problem, where an intelligent, pre-trained model at the receiver is utilized to infer the real-time value of a target signal, $Y_t$. This inference is based on the most recently received data packet, denoted as ${X_{t-\delta}^l = (V_{t-\delta}, V_{t-\delta-1}, \ldots, V_{t-\delta-l+1})}$, which has been transmitted from a remote location. The data packet $X_{t-\delta}^l$ comprises samples of the source signal $V_t$ at the transmitter. Here, $\delta$ represents the AoI of the freshest sample in the packet, and $l$ denotes the packet length, corresponding to the number of samples included in the packet.

The inference error is characterized as a function of both the AoI and the packet length. Recent studies \cite{shisher2022does, shisher2024timely} have demonstrated that, for various remote inference systems, the relationship between the inference error and AoI for a given packet length can be non-monotonic. This finding contradicts the conventional assumption that a packet with ${\text{AoI} = 0}$ necessarily yields better performance than one with ${\text{AoI} > 0}$. Consequently, rather than adopting the generate-at-will model~\cite{yates2015lazy, sun2017update}, we use the ``selection-from-buffer'' model proposed in~\cite{shisher2022does, shisher2024timely}, which enables the scheduler to choose either fresh or stale data samples from the buffer.

Such non-monotonicity between the inference error and AoI arises when the source signal, $V_t$, and the target signal, $Y_t$, exhibit a delayed relationship, such as $Y_t = f(V_{t-\Upsilon})$ \cite{shisher2024timely, shisher2021age, shisher2024monotonicity}. For instance, $\Upsilon$ may represent a combination of the communication delay between the controller and the actuator and the actuation delay in a physical networked control system. An example of this scenario is illustrated in the leader-follower robot setup in \cite[Fig.~2]{shisher2024timely}. Additionally, non-monotonicity also occurs when the relationship between $V_t$ and $Y_t$ exhibits periodicity, as shown in the temperature prediction example in \cite[Fig.~4]{shisher2024timely}. In such cases, when communication delays prevent the transmission of sufficiently fresh samples to the receiver, transmitting aged samples can exploit the periodicity to enhance performance.

The inference error is non-increasing with packet length for a given AoI~\cite{shisher2023learning}, as additional information can only improve the estimation accuracy. However, this improvement comes at the cost of longer transmission delays. In our communication design, we investigate the interplay between packet length and transmission delay, along with the possibly non-monotonic relationship between inference error and AoI.

Finally, next-generation communication networks are expected to accommodate numerous connections, resulting in multiple routes between any pair of nodes. The relays in such networks may enforce specific routes for certain data flows, which can dynamically change over time to address the varying demands of multiple flows. For instance, consider a network involving both terrestrial and non-terrestrial connections. A flow may be serviced entirely by non-terrestrial connections, terrestrial connections, or a combination of both, depending on factors such as the flow's priority, network congestion, or satellite availability. Therefore, goal-oriented communication design must be adaptable to delay conditions that vary significantly with memory.

To that end, the technical contributions of this paper are:
\begin{itemize}
    \item We extend the system models in~\cite{shisher2024timely, shisher2023learning} to more practical scenarios by considering Markovian delay on both forward and feedback links. Based on this extension, we formulate and solve a learning and communication co-design problem for real-time remote inference. We consider both time-invariant and time-variable packet length selection to address the varying computational capabilities of practical systems. The derived optimal packet selection and scheduling policies minimize the time-average inference error for a given, possibly non-monotonic, inference error function corresponding to a particular remote inference application, thereby making the communication \textit{goal-oriented}.
    \item We formulate the learning and communication co-design problem under time-invariant packet length selection as a two-layer nested optimization problem. The inner layer determines the freshness and transmission times of data packets for a given constant packet length~$l$ and is modeled as an infinite-horizon average-cost Semi-Markov Decision Process (SMDP). Such problems are often solved using dynamic programming \cite{puterman2014markov, bertsekasdynamic} and typically do not have closed-form solutions. However, we derive a closed-form solution to this inner layer problem, presented in Theorem~\ref{Theorem 1}. The solution consists of two main parts: (i) the freshness of the data packets to be transmitted is determined by a stationary function of the network state, i.e., the delay state, in the previous epoch, and (ii) the transmission instants are controlled by an index-based threshold policy. The index function is expressed in terms of the delay state of the network in the previous epoch and the AoI at the receiver. We then separately optimize the value of the constant packet length~$l$ in the outer layer.
    \item We model the learning and communication co-design problem under time-variable packet length selection directly as a single infinite-horizon average-cost SMDP. We then formulate a Bellman optimality equation and derive a structural result regarding the optimal solution. Motivated by this result, we present a simplified version of the Bellman optimality equation in Theorem \ref{Theorem 2}. Solving the simplified Bellman equation using dynamic programming has significantly lower time complexity compared to solving the original Bellman optimality equation.
    \item We conduct two experiments to evaluate the performance of our optimal scheduling policies: (i) remote inference of auto-regressive (AR) processes, offering a model-based evaluation; and (ii) cart-pole state prediction, providing a trace-driven evaluation. In simulation results, we observe that our goal-oriented scheduler drops inference error down to one-sixth with respect to the age-based scheduling of unit-length packets.
\end{itemize}

\subsection{Related Work}
The concept of Age of Information (AoI) has attracted significant research interest; see, e.g., \cite{shisher2021age,kaul2012real,sun2017update, yates2015lazy, kadota2018optimizing, sun2019sampling2, ornee2021sampling, tripathi2019whittle, klugel2019aoi, sun2019closed, ornee2023whittle, pan2023sampling, ornee2023context, shisher2023learning, shisher2024monotonicity, shisher2024timely, shisher2022does, bedewy2020optimizing, bedewy2021optimal} and a recent survey \cite{yates2021age}. Initially, research efforts were centered on analyzing and optimizing the average AoI and peak AoI in communication networks \cite{kaul2012real, sun2017update, yates2015lazy, kadota2018optimizing}. Recent research endeavors have revealed that the performance of real-time applications can be modeled as non-linear functions of AoI, leading to the study of optimizing these non-linear functions in control system scenarios~\cite{klugel2019aoi, soleymani2019stochastic}, remote estimation~\cite{sun2019sampling, ornee2021sampling, pan2023sampling}, and remote inference~\cite{shisher2021age, shisher2022does, shisher2023learning, shisher2024timely}. 

While many studies have analyzed AoI in queuing models, closest to the spirit of this paper is the control of AoI via replacement of exogenous data arrivals with the generation of data "at will" \cite{bacinoglu2015age, sun2017update, kadota2018scheduling, kadota2019scheduling, sun2019sampling2, tripathi2019whittle, klugel2019aoi, sun2019closed}. A generalization of this approach is to incorporate jointly optimal sampling and scheduling policies to control not only AoI but a more sophisticated end-to-end distortion criterion by using AoI as an auxiliary parameter \cite{sun2019sampling, ornee2021sampling,  bedewy2020optimizing, shisher2023learning, shisher2022does, shisher2024timely}. While strikingly more demanding of analysis, these formulations take us closer to goal-oriented communication design.

Almost all previous studies on the ``generate-at-will'' model adopted an assumption that the penalty of information aging is a non-decreasing function of the AoI \cite{sun2019sampling, ornee2021sampling, sun2017update, kadota2018scheduling, kadota2019scheduling, sun2019sampling2, tripathi2019whittle, klugel2019aoi, sun2019closed, bedewy2021optimal, bedewy2020optimizing}. However, it was shown in \cite{shisher2021age, shisher2022does, shisher2023learning, shisher2024monotonicity, shisher2024timely} that the monotonicity of information aging depends heavily on the divergence of the time-series data from being a Markov chain. If the input and target data sequences in a system can be closely approximated as a Markov chain, then the penalty increases as the AoI grows; otherwise, if the data sequence significantly deviates from a Markovian structure, the monotonicity assumption does not hold. Following the approach in \cite{shisher2021age, shisher2022does, shisher2023learning, shisher2024monotonicity, shisher2024timely}, this paper models the inference error as a possibly non-monotonic function of AoI.

The works most closely related to this paper are~\cite{shisher2024timely, shisher2023learning}, which developed scheduling policies for remote inference, considering a possibly non-monotonic dependency between AoI and practical performance. While \cite{shisher2024timely} developed scheduling policies for fixed packet length, \cite{shisher2023learning} jointly optimized packet length and scheduling strategies. However, both studies~\cite{shisher2024timely, shisher2023learning} assumed random i.i.d. delay for packet transmissions and immediate feedback. This paper extends the problem formulations in~\cite{shisher2024timely, shisher2023learning} to more practical scenarios and develops a \textit{goal-oriented} communication strategy that minimizes the average inference error for remote inference under two-way delay (i.e., incorporating random feedback delay~\cite{tsai2021unifying, pan2022optimal}) that varies significantly with memory.

The earlier version of this paper~\cite{ari2023goal} is a special case of the learning and communication co-design problem under time-invariant packet length selection discussed in Section~\ref{sec_time_inv}. In \cite{ari2023goal}, the scheduler is restricted to transmitting a single sample at each transmission, i.e., the packet length is fixed to $l=1$. In contrast, the current paper allows the scheduler to form packets comprising multiple consecutive samples for transmission. Packets containing a larger number of samples can yield better inference performance. However, the associated increase in transmission delay makes such packets staler upon arrival, which may eventually degrade inference performance. As a result, this paper captures the trade-off between the volume of data transmitted and the resulting increase in transmission delay while optimizing the task performance at the destination. 

Furthermore, \cite{ari2023goal} established the optimal threshold-based waiting-time rule for the time-invariant packet length $l=1$ and showed that the freshness of the data packets to be transmitted is a function only of the network delay state in the previous epoch. The present paper extends these results by generalizing the waiting-time rule to arbitrary time-invariant packet lengths~$l$ (not limited to $l=1$) and by deriving an exact closed-form expression for the freshness of the data packets to be transmitted. Finally, this paper also investigates scenarios where the packet length can be adjusted over time based on the state of the scheduler in Section~\ref{sec_4}.

\section{System Model} \label{sys_model}

\begin{figure}[t]
\centerline{\includegraphics[width=.51\textwidth]{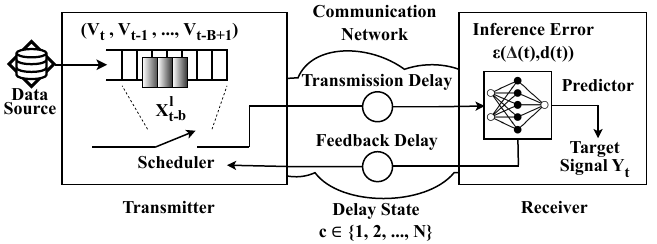}}
\caption{A remote inference system that adopts the selection-from-buffer medium access model \cite{shisher2024timely, shisher2022does}.}
\label{System_Model}
\end{figure}

We consider a time-slotted remote inference system, which comprises a data source, a transmitter, and a receiver, as illustrated in Fig.~\ref{System_Model}. On the transmitter side, the source signal $V_t$ is regularly sampled in each time slot $t$. The samples are stored in a buffer containing the most recent $B$ samples $\big(V_{t}, V_{t-1}, \ldots, V_{t-B+1}\big)$, waiting to be sent to the receiver. On the receiver side, a predictor, e.g., a pre-trained neural network, is used to infer the real-time value of the target signal $Y_t$ of interest. In practical time-series forecasting algorithms, the predictor needs to use multiple consecutive samples, e.g., $X_{t-b}^l = (V_{t-b}, V_{t-b -1}, \ldots, V_{t-b-l+1})$, to infer $Y_t$. We will design a transmission scheduler that determines (i) when to send a group of consecutive samples over the communication network to the receiver and (ii) which group of consecutive samples in the buffer to send. This medium access model is called ``selection-from-buffer,'' which was proposed recently in \cite{shisher2022does, shisher2024timely}. Such remote inference systems are crucial for numerous real-time applications, including sensor networks, airplane/vehicular control, robotics networks, and digital twins, where timely and accurate remote monitoring and control are essential for maintaining system stability and performance. 

We assume that the system starts operation at time slot $t = 0$. In the selection-from-buffer model, a new sample of the source signal $V_{t} \in \mathcal{V}$ is added to the buffer at the beginning of each time slot $t$, and the oldest sample is discarded. The buffer, with size $B$, is assumed to be initially full, storing the samples $\big(V_{0}, V_{-1}, \ldots, V_{-B+1}\big)$. Thus, at any time slot $t$, the buffer contains the most recent $B$ samples $\big(V_{t}, V_{t-1}, \ldots, V_{t-B+1}\big)$. The transmitter sends status update packets to the receiver one by one. Let $S_i$ and $D_i$ represent the transmission start time and delivery time of the $i$-th transmitted packet, respectively, such that $S_i < D_i$. At time slot $t=S_i$, the scheduler selects a group of consecutive samples from the buffer and forms the packet $X_{S_i-b_i}^{l_{i}} = \big(V_{S_i-b_{i}}, V_{S_i-b_{i}-1}, \ldots, V_{S_i-b_{i}-l_{i}+1}\big) \in \mathcal{V}^{l_i}$. Here, $l_i \in \{1, \ldots, B\}$ specifies the sequence length of the $i$-th packet $X_{S_i-b_i}^{l_{i}}$, while $b_i$ denotes the relative position of the sample $V_{S_i-b_{i}}$ in the buffer. The latter parameter $b_i$ controls the freshness of the samples included in the packet. When $B=1$, the selection-from-buffer model reduces to the generative-at-will model which has been studied extensively in previous works \cite{bacinoglu2015age, sun2017update, kadota2018scheduling, kadota2019scheduling, sun2019sampling2, tripathi2019whittle, klugel2019aoi, sun2019closed}. Hence, the packet $X_{S_i-b_i}^{l_{i}}$ is submitted to the communication network at time slot $S_{i}$ and delivered to the receiver at time slot $D_{i}$. Upon delivery, the receiver sends an acknowledgment (ACK) back to the transmitter, which is received at time slot $A_{i}$.

We assume that the scheduler always waits for feedback before transmitting a new data packet. In other words, the scheduler remains silent during the time slots between $D_{i}$ and $A_{i}$ such that $S_i < D_i < A_i \leq S_{i+1}$ for all~$i$. This assumption clearly restricts the throughput, the maximization of which has been the main objective in conventional user-centric, pipelined communication systems. However, under the concept of goal-oriented communication, which has been introduced to address the scalability problem in massive machine-type networks, the main objective is not throughput maximization but ensuring the effective execution of machine tasks at the destination while operating under limited communication resources. From this perspective, allowing multiple packets to be simultaneously in flight may increase network load without yielding proportional improvements in task performance, particularly in dense environments where congestion and contention are the primary concerns. Consequently, the one-packet-in-flight (stop-and-wait) assumption serves as a scenario-driven modeling choice that captures conservative communication behavior suitable for the class of systems considered in this work.

\subsection{Communication Network Model}

Next-generation communication networks are expected to accommodate multiple routes between a transmitter–receiver pair through a variety of intermediate network components, such as LEO satellites, inter-satellite links, and ground gateways. The experienced delay may therefore vary significantly depending on the specific components traversed by a packet and the manner in which it is handled during transmission. For example, when a packet is delivered entirely via inter-satellite links under moderate network load, the dominant delay component is propagation delay. In contrast, when inter-satellite link unavailability requires routing through intermediate ground stations, or when packets undergo store-and-forward transmission over long distances in LEO satellites under heavy network load, the resulting delay may involve either substantially larger propagation components or considerable store-and-forward waiting times. In such networks, the receiver can characterize the realized end-to-end delay behavior through a detection rule that is consistent with its observation capabilities (e.g., based on measured packet arrival times). The precision of this characterization depends on the receiver’s observation capability.

Furthermore, as in any practical communication system, physical-layer transmission errors and link-layer retransmissions may occur along these routes. In this work, we adopt a high-level network-layer perspective, where such effects are incorporated into the end-to-end transmission delay. Specifically, the delay experienced by a packet accounts for processing and propagation delays, as well as potential local retransmissions and error recovery mechanisms within the network. This modeling choice allows us to focus on the scheduler’s decision-making at the source, rather than on lower-layer protocol operations.

\begin{figure}[t]
\centerline{\includegraphics[width=.51\textwidth]{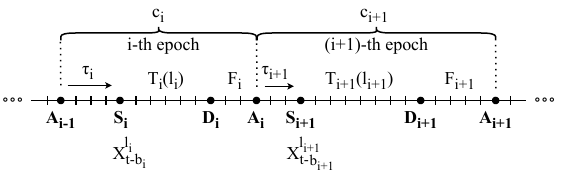}}
\caption{An illustration of the $i$-th and the $(i+1)$-th packet transmission epochs.}
\label{Time_Evolution}
\end{figure}

Let the $i$-th packet transmission epoch consist of the time slots $A_{i-1}, A_{i-1}+1, \ldots, A_i-1$. Define $T_i(l_i) = D_i - S_i$ as the transmission delay incurred in the $i$-th epoch and $F_i = A_i - D_i$ as the feedback delay. In addition, let $C_i$ denote the network delay state in the $i$-th epoch, with its realization~$c_i$. Fig.~\ref{Time_Evolution} provides an illustration of the $i$-th and the $(i+1)$-th packet transmission epochs.

The network delay state $C_i$ in the $i$-th epoch can take values in $\{1, 2,\ldots, N\}$, each representing a distinct delay regime that the network may experience during the transmission of the $i$-th data packet and that can be distinguished by the receiver. Each delay regime is associated with different transmission and feedback delay distributions, and the transmission delay also varies with the packet length. The delay is at least one time slot for each packet or ACK transmission. 

We assume that $C_{i}$ evolves according to a finite-state ergodic Markov chain with transition probabilities $p_{ij}$, where $i,j \in \{1, 2,\ldots, N\}$ and $p_{ij}$ represents the probability of transitioning from state~$i$ to state~$j$. The Markov chain makes a single transition at time slot $A_{i}$ and none otherwise.

The receiver detects the network delay state $C_i$ at the packet delivery time $D_i$ and embeds this information into the ACK message. The ACK notifies the transmitter of (i) the delay state experienced during the previous packet transmission, (ii) the completion of the previous packet transmission, and (iii) that the next packet can be submitted to the network. The network is reliable, meaning no packet is lost during transmission.

\subsection{Inference Error as a Function of AoI and Packet Length}

Age of Information (AoI) on the receiver side, $\Delta(t)$, is the time difference between the current time $t$ and the generation time $t-\Delta(t)$ of the freshest sample $V_{t-\Delta(t)}$ in the most recently delivered packet $X_{t-\Delta(t)}^{l}$ \cite{kaul2012real}. The AoI at time slot $t$ is determined by
\begin{equation} \label{AoI Evolution}
\Delta(t) = t-S_{i}+b_{i}, \text{  if  } D_{i} \leq t < D_{i+1}. 
\end{equation} 

Let $d(t)$ represent the sequence length of the most recently delivered packet by time slot $t$, which is given by
\begin{equation} \label{length of the most recently delivered packet}
d(t) = l_i, \text{  if  } D_{i} \leq t < D_{i+1}.
\end{equation}

We assume that the predictor employs a supervised learning algorithm based on \textit{Empirical Risk Minimization}~\cite{goodfellow2016deep}. For each packet length $l \in \{1, 2, \ldots, B\}$, the receiver contains a pre-trained neural network $\phi^*_{l}:\mathbb Z^+\times\mathcal{V}^l\mapsto\mathcal{A}$, which takes the AoI $\Delta(t)$ and the packet $X_{t-\Delta(t)}^{l}$ as inputs and produces an action~$a$ to infer the real-time value of the target signal~$Y_t$. The performance is evaluated through a loss function ${L:\mathcal{Y}\times\mathcal{A}\mapsto\mathbb{R}}$. The incurred loss is $L(y, a)$ when the action~$a$ is taken to predict $Y_t = y$. Given the AoI $\delta$ and packet length $l$, the training problem to obtain the neural network $\phi^*_{l}$ is formulated as:
\begin{align}
\phi^*_{l} = \argmin_{\phi \in \Lambda} \mathbb{E}_{Y,X^l \sim P_{\tilde Y_0, \tilde X_{-\delta}^l}}\big[L\big(Y,\phi\big(\delta, X^l\big)\big)\big],
\end{align}
where $\Lambda$ is the set of all mappings that the neural network can generate, and $P_{\tilde Y_0, \tilde X_{-\delta}^l}$ represents the empirical distribution of the target signal $\tilde Y_0$ and the packet $\tilde X_{-\delta}^l$ in the training dataset. The AoI $\delta$ denotes the time difference between the generation~of $\tilde Y_0$ and $\tilde X_{-\delta}^l$.

These pre-trained neural networks are used at the destination to predict the real-time value of the target signal $Y_t$ based on the most recently received data packet. We assume that the stochastic process ${\{\big(Y_t, V_t\big), t = 0, 1, \ldots\}}$ is stationary, implying that the statistical relationship between the source signal and the target signal is time-invariant. In addition, we further assume that the process ${\{\big(Y_t, V_t\big), t = 0, 1, \ldots\}}$ is independent of the process ${\{\big(\Delta(t), d(t)\big), t = 0, 1, \ldots\}}$. Accordingly, the average inference error at time slot $t$, conditioned on $\Delta(t)=\delta$ and $d(t)=l$, is expressed as:
\begin{equation} \label{Inference Error}
    \varepsilon\big(\delta,l\big) = \mathbb{E}_{Y, X^{l}  \sim P_{Y_{t}, X_{t-\delta}^{l}}}\big[L\big(Y,\phi^*_{l}\big(\delta, X^{l}\big)\big)\big],
\end{equation}
where $P_{Y_{t}, X_{t-\delta}^{l}}$ represents the joint distribution of the target signal $Y_{t}$ and the packet $X_{t-\delta}^{l}$.

The loss function $L$ can be chosen based on the goal of the remote inference application. For example, a quadratic loss function $L_{2}\big(y,\hat{y}\big)=\lVert y-\hat{y} \rVert_{2}^{2}$ is used in neural network-based minimum mean-squared estimation, where the action $a=\hat{y}$ is an estimate of the target signal $Y_{t}=y$ and $\lVert y \rVert_{2}$ is the Euclidean norm. In softmax regression (i.e., neural network-based maximum likelihood classification), the action $a=\mathcal{Q}_{Y}$ is a distribution of $Y_{t}$, and the loss function $L_{\log}\big(y,\mathcal{Q}_{Y}\big)=-\log\mathcal{Q}_{Y}(y)$ is the negative log-likelihood function of~the~value~$Y_{t}=y$.

In the subsequent sections, we solve a learning and communication co-design problem that aims to optimize the performance of a remote inference system. From equation \eqref{Inference Error}, we know that the inference error is a function of both the AoI and packet length. For a given AoI $\delta$, the inference error is a non-increasing function of packet length \cite{shisher2023learning}. On the other hand, for a given packet length $l$, the inference error does not always increase monotonically with the AoI \cite{shisher2022does, shisher2024timely}. When minimizing non-decreasing functions of AoI, fresher data packets are preferable, and the generate-at-will medium access model works well. However, to minimize inference error, which is not necessarily non-decreasing with respect to AoI, the selection-from-buffer model—allowing for the choice between fresh and older data packets—is more effective.

\subsection{Structure of the Scheduling Policy}
Upon receiving the ACK at time slot $A_i$, the scheduler determines a waiting time $\tau_{i+1}$ to specify the next submission time slot as $S_{i+1}=A_{i}+\tau_{i+1}$. Then, a new data packet $X_{S_{i+1}-b_{i+1}}^{l_{i+1}}$ is formed with packet length $l_{i+1}$ and relative buffer position $b_{i+1}$. A scheduling policy $\pi \in \Pi$ is defined as a tuple $\pi = \big(f, \ell, g\big)$, where $f = \big(b_{2}, b_{3}, \ldots\big)$ is the relative buffer position sequence, $\ell = \big(l_{2}, l_{3}, \ldots\big)$ is the packet length sequence, $g = \big(\tau_{2}, \tau_{3}, \ldots\big)$ is the waiting time sequence. In addition, let $\Pi$ denote the set of all causal policies $\pi$. We assume that the scheduler does not utilize knowledge of the process $\{\big(Y_t, V_t\big), t = 0, 1, \ldots\}$. This assumption implies that the processes $\{\big(Y_t, V_t\big), t = 0, 1, \ldots\}$ and $\{\big(\Delta(t), d(t)\big), t = 0, 1, \ldots\}$ are mutually independent. The initial conditions of the system are assumed to be $S_{1}=0$, $l_{1}=1$, $b_1 = 0$, $c_{1} \in \{1, 2, \ldots, N\}$, and $\Delta(0)$ is a finite constant. 

\section{Learning and Communications Co-design: Time-invariant Packet Length Selection} \label{sec_time_inv}
While designing goal-oriented communication protocols for remote inference, selecting a constant packet length is a simple and often effective choice. If packet lengths vary dynamically from one packet to the next, the receiver needs to reconfigure the predictor over time to accommodate the changing input dimensions, which adds complexity. This section considers systems where this complexity is avoided by choosing a time-invariant packet length sequence. Systems that can vary the packet length over time will be discussed later in Section \ref{sec_4}.

\subsection{Co-design Problem Formulation}
Let $\Pi_{l}$ represent the set of all causal scheduling policies $\pi=(f, \ell, g)$ with a time-invariant packet length sequence ${\ell = \{l,l, \ldots\}}$ for $l \in \{1, 2, \ldots, B\}$:
\begin{equation}
    \Pi_l = \{\pi\in \Pi: l = l_2 = l_3 = \ldots \}.
\end{equation}
The packet length affects both learning and communication. If it is selected based solely on learning performance, communication may suffer since transmission delays increase with packet length. Therefore, we solve a learning and communications co-design problem under time-invariant packet length selection, formulated as the following two-layer nested optimization problem.
\begin{equation} \label{Problem 1: Inner opt problem}
    \varepsilon_{l, \text{opt}}= \inf_{\pi \in \Pi_{l}}\limsup_{T\to\infty} \frac{1}{T}\mathbb{E}_{\pi}\Bigg[\sum_{t=0}^{T-1} \varepsilon\big(\Delta(t), l\big)\Bigg],
\end{equation} 
\begin{equation} \label{Problem 1: Outer opt problem}
    \varepsilon_{\text{opt}}^{\text{inv.}} = \min_{l \in \{1, 2, \ldots, B\}} \varepsilon_{l, \text{opt}},
\end{equation} where $\varepsilon_{l, \text{opt}}$ is the optimum value of the inner optimization problem \eqref{Problem 1: Inner opt problem} for a given packet length $l$, and $\varepsilon_{\text{opt}}^{\text{inv.}}$ is the optimum value of the two-layer nested optimization problem \eqref{Problem 1: Inner opt problem}-\eqref{Problem 1: Outer opt problem} with the optimal time-invariant packet length.

The inner problem \eqref{Problem 1: Inner opt problem} optimizes the relative buffer position sequence $f$ and waiting time sequence $g$ for a given packet length $l \in \{1, 2, \ldots, B\}$. This problem can be cast as an infinite-horizon average-cost Semi-Markov Decision Process (SMDP). Although dynamic programming algorithms, such as policy iteration and value iteration, are typically used to solve such problems~\cite{puterman2014markov, bertsekasdynamic}, we are able to derive a closed-form solution that achieves significantly lower computational complexity than dynamic programming. Moreover, the outer problem \eqref{Problem 1: Outer opt problem} optimizes the packet length and involves a straightforward search over the integer values \(\{1, 2, \ldots, B\}\).

Problem \eqref{Problem 1: Inner opt problem}-\eqref{Problem 1: Outer opt problem} is a more general formulation than that considered in \cite[Section IV-A]{shisher2023learning}, as it accounts for both (i) the two-way delay between the transmitter and receiver, and (ii) the Markovian, time-varying nature of the two-way delay distribution in next-generation communication networks.

\subsection{Optimal Solution to \eqref{Problem 1: Inner opt problem}-\eqref{Problem 1: Outer opt problem}} \label{SMDP_1}
Each time slot \(A_i\) is a decision time for the SMDP \eqref{Problem 1: Inner opt problem}. The system state at each $A_i$ is given by a tuple $\big(\Delta(A_i), C_i\big)$, where $\Delta(A_i)$ represents the AoI at time-slot $A_i$, and $C_i$ indicates the delay state in the $i$-th packet transmission epoch. A realization of the system state is denoted by $\big(\delta, c_i\big)$. The actions determined at time slot $A_i$ are the waiting time \(\tau_{i+1}\) and the buffer position \(b_{i+1}\). A detailed description of this SMDP can be found in Appendix \ref{App_Comp_SMDP_1}. The Bellman optimality equation of the SMDP \eqref{Problem 1: Inner opt problem} is formulated as
\begin{align}\label{Bellman Equation for fixed packet length}
&h_1\big(\delta, c_i\big)= \min_{\substack{\tau_{i+1} \in \{0, 1, \ldots\} \\ b_{i+1} \in \{0,1,\ldots,B-l\}}} \nonumber\\
&\mathbb E \Bigg[ \sum_{k=0}^{\tau_{i+1}+T_{i+1}(l)-1} \bigg(\varepsilon\big(\delta+k,l\big) -\varepsilon_{l, \text{opt}}\bigg) \Bigg| C_i=c_i\Bigg]\nonumber\\
&+\mathbb E \Bigg[ \sum_{k=0}^{F_{i+1}-1} \bigg(\varepsilon\big(b_{i+1}+T_{i+1}(l)+k,l\big) -\varepsilon_{l, \text{opt}}\bigg) \Bigg| C_i=c_i\Bigg]\nonumber\\
&+\mathbb E\big[h_1\big(b_{i+1}+T_{i+1}(l)+F_{i+1}, C_{i+1}\big)\big| C_i=c_i\big],
\end{align}
where $\delta \in \mathbb Z^+$, $c_i \in \{1, 2, \ldots, N\}$, and $h_1(\delta, c_i)$ is the relative value function.

We need to solve the Bellman optimality equation \eqref{Bellman Equation for fixed packet length} to find an optimal solution to \eqref{Problem 1: Inner opt problem}. This involves determining the optimal buffer position $b_{i+1}^*$ and waiting time $\tau_{i+1}^*$ for each state $\big(\delta, c_i\big)$. By exploiting the structural properties of the Bellman optimality equation \eqref{Bellman Equation for fixed packet length}, we are able to obtain a closed-form solution to \eqref{Problem 1: Inner opt problem}, as asserted in the following theorem.

Define the index function
\begin{equation} \label{index_function_const_packet_length}
    \gamma\big(\delta, d, l, c\big)=\min_{\tau\in\{1,2,\ldots\}} \frac{1}{\tau}\sum_{k=0}^{\tau-1}\mathbb{E}\bigg[\varepsilon\big(\delta+T_{i+1}(l)+k,d\big)\bigg|C_{i}=c\bigg],
\end{equation} for all $\delta \in \mathbb Z^+$, $d, l \in \{1, 2,\ldots, B\}$, and $c \in \{1, 2,\ldots, N\}$.

\begin{theorem} \label{Theorem 1}
    There exists an optimal solution $\pi^* = \big(\big(b_2^*, b_3^*, \ldots\big), \big(l, l, \ldots\big), \big(\tau_2^*, \tau_3^*, \ldots\big)\big) \in \Pi_l$ to problem \eqref{Problem 1: Inner opt problem}, where the optimal waiting time $\tau_{i+1}^*$ is determined by the index-based threshold rule
    \begin{equation} \label{Theorem_1_Wait_Rule}
        \tau_l\big(\delta, c_i\big) = \min\{ k \geq 0  : \gamma\big(\delta+k, l, l, c_i\big) \geq \beta\},
    \end{equation} the optimal buffer position $b_{i+1}^*$ is given by
    \begin{align} \label{Theorem_1_Buf_Pos_Seq}
        &b_{i+1}^* = \argmin_{b \in \{0,1,\ldots,B-l\}} \nonumber \\
        &\mathbb{E} \Bigg[ \sum_{k=0}^{\mathfrak{D}-1} \bigg(\varepsilon\big(b+T_{i+1}(l)+k,l\big) - \beta \bigg) \Bigg| C_i=c_i\Bigg],
    \end{align} $\mathfrak{D} = F_{i+1}+\tau_l\big(b+T_{i+1}(l)+F_{i+1}, C_{i+1}\big)+T_{i+2}(l)$, and $\gamma\big(\cdot\big)$ is defined in \eqref{index_function_const_packet_length}. Furthermore, $\beta$ is the unique root of
    \begin{equation} \label{Theorem_1_unique_root_of}
        \mathbb{E}\Bigg[\sum_{t=A_{i}(\beta)}^{A_{i+1}(\beta)-1} \varepsilon\big(\Delta(t), l\big)\Bigg]-\beta\mathbb{E}\big[A_{i+1}(\beta)-A_{i}(\beta)\big] =0,
    \end{equation} and $\beta = \varepsilon_{l, \text{opt}}$ is exactly the optimum value of problem~\eqref{Problem 1: Inner opt problem}. $A_{i}(\beta)$ in equation \eqref{Theorem_1_unique_root_of} represents the $i$-th ACK reception time slot when the scheduling policy is determined according to the rules specified by equations~\eqref{Theorem_1_Wait_Rule} and \eqref{Theorem_1_Buf_Pos_Seq} for a given $\beta$.
\end{theorem}

\begin{proof}[Proof Sketch]
    We prove Theorem \ref{Theorem 1} in three steps:
    
    \textbf{Step 1:} We first establish two key results using the Bellman optimality equation \eqref{Bellman Equation for fixed packet length}:
    \begin{itemize}
        \item The optimal waiting time $\tau_{i+1}^*$ is determined by the index-based threshold rule 
        $$
        \min\{k \geq 0 : \gamma\big(\delta+k, l, l, c_i\big) \geq \varepsilon_{l, \text{opt}}\}
        $$ 
        for all $i$. Here, $\gamma\big(\cdot\big)$ is the index function defined in \eqref{index_function_const_packet_length}.
        \item The delay state \(c_i\) in the $i$-th epoch is a sufficient statistic for determining the optimal buffer position \(b_{i+1}^*\) for all $i$.
    \end{itemize}
    
    \textbf{Step 2:} Since the delay state $c_i$ in the $i$-th epoch is a sufficient statistic for determining $b_{i+1}^*$ for all $i$, we can construct a new SMDP for determining $b_{i+1}^*$ using $c_i$ as the state. In the new SMDP, we fix the waiting time decisions by using the optimal threshold rule $\tau_l\big(\delta, c_i\big)$. Moreover, in the new SMDP, the decision time is $D_{i+1}$ instead of $A_i$.  The Bellman optimality equation of this new SMDP is expressed as:
    \begin{align}\label{New Bellman Equation for fixed packet length}
    &h_1'\big(c_i\big)= \min_{\substack{b_{i+1} \in \{0,1,\ldots,B-l\}}} \nonumber\\
    &\mathbb E \Bigg[ \sum_{k=0}^{\vartheta_{i+1}-1} \bigg(\varepsilon\big(b_{i+1}+T_{i+1}(l)+k,l\big) -\varepsilon_{l, \text{opt}}\bigg) \Bigg| C_i=c_i\Bigg]\nonumber\\
    &+\mathbb E\bigg[h_1'\big(C_{i+1}\big)\bigg| C_i=c_i\bigg],
    \end{align} where $c_i \in \{1, 2, \ldots, N\}$, $h_1'\big(c_i\big)$ is the relative value function, and $\vartheta_{i+1} = F_{i+1}+\tau_l\big(b_{i+1} + T_{i+1}(l) + F_{i+1}, C_{i+1}\big)+T_{i+2}(l)$.
    
    Only the first term on the right-hand side of equation \eqref{New Bellman Equation for fixed packet length} depends on the action $b_{i+1}$. Thus, the Bellman optimality equation \eqref{New Bellman Equation for fixed packet length} is decomposable and can be solved as a per-decision-epoch optimization problem, allowing any $b_{i+1}^*$ in the optimal buffer position sequence ${f^* = \big(b_2^*, b_3^*, \ldots\big)}$ to be expressed as shown in equation \eqref{Theorem_1_Buf_Pos_Seq}.
    
    \textbf{Step 3:} Finally, we show that the optimal value $\varepsilon_{l, \text{opt}}$ of problem \eqref{Problem 1: Inner opt problem} is the unique root of equation \eqref{Theorem_1_unique_root_of}.
    
    The detailed proof is provided in Appendix \ref{App_Proof_TH1}.
\end{proof}

The optimal scheduling policy for the inner optimization problem \eqref{Problem 1: Inner opt problem}, as outlined in Theorem \ref{Theorem 1}, is well-structured. Each optimal waiting time \(\tau_{i+1}^*\) is determined by an index-based threshold rule \(\tau_l\big(\delta, c_i\big)\). The index function \(\gamma\big(\cdot\big)\), defined in equation~\eqref{index_function_const_packet_length}, can be readily computed for any state \(\big(\delta, c_i\big)\) using the conditional distribution of \(T_{i+1}(l)\) given \(C_i = c_i\). Although the index-based threshold rule $\tau_l\big(\delta, c_i\big)$ closely resembles the waiting time rule provided in \cite[Theorem 1]{shisher2023learning}, it is not solely dependent on the AoI but also incorporates the delay state $C_i = c_i$. Because problem \eqref{Problem 1: Inner opt problem}-\eqref{Problem 1: Outer opt problem} involves non-zero feedback delay and time-varying delay distribution, the existence of such a simple waiting time rule was not evident and required additional technical efforts for its derivation.

Furthermore, according to equation \eqref{Theorem_1_Buf_Pos_Seq}, the optimal buffer position \(b_{i+1}^*\) is designed by minimizing the relative inference error in the interval between \(D_{i+1}\) and \(D_{i+2}\), given that the waiting time is determined by the optimal rule in \eqref{Theorem_1_Wait_Rule}. Each \(b_{i+1}^*\) depends solely on the delay state $C_i = c_i$ in the $i$-th epoch and is independent of the AoI state \(\delta = \Delta(A_i)\). 

For the case of i.i.d. transmission delay and immediate feedback, the optimal buffer position sequence is time-invariant, as shown in \cite[Theorem 1]{shisher2023learning}. However, when the feedback delay is non-zero and the delay state $C_i$ changes according to a Markov chain, achieving optimal performance requires that the optimal buffer position $b_{i+1}^*$ is a function of the delay state $C_i = c_i$, and this function is exactly given by \eqref{Theorem_1_Buf_Pos_Seq}.

Finally, the optimum value \(\varepsilon_{l, \text{opt}}\), which is the threshold in the waiting time rule \eqref{Theorem_1_Wait_Rule} and necessary for solving \eqref{Theorem_1_Buf_Pos_Seq}, is the unique root of equation \eqref{Theorem_1_unique_root_of}. This equation can be efficiently solved using low-complexity algorithms \cite[Algorithms~1-3]{ornee2021sampling}, such as bisection search.

After solving the inner optimization problem \eqref{Problem 1: Inner opt problem} using Theorem \ref{Theorem 1}, the outer optimization problem \eqref{Problem 1: Outer opt problem} is just a straightforward search over the integer values $\{1, 2, \ldots, B\}$.

\section{Learning and Communications Co-design: Time-Variable Packet Length Selection} \label{sec_4}

In contrast to the previous case of time-invariant packet length, this section aims to achieve optimal performance without imposing restrictions on the packet length sequence $\ell = \big(l_2, l_3, \ldots\big)$. This unconstrained approach allows us to maximize the system's potential by dynamically adjusting the packet length $l_i$ over time. Therefore, this case represents the most general solution.

\subsection{Co-design Problem Formulation}

The learning and communications co-design problem with general, time-variable packet length selection is formulated as
\begin{equation} \label{Problem 3}
    \varepsilon_{\text{opt}}=\inf_{\pi \in \Pi}\limsup_{T\to\infty} \frac{1}{T}\mathbb{E}_{\pi}\Bigg[\sum_{t=0}^{T-1} \varepsilon\big(\Delta(t), d(t)\big)\Bigg],
\end{equation} where $\varepsilon_{\text{opt}}$ is the optimum value of problem \eqref{Problem 3}.

Problem \eqref{Problem 3} is an infinite-horizon average-cost SMDP, which is more challenging to solve compared to problem~\eqref{Problem 1: Inner opt problem}. We will first formulate a Bellman optimality equation for this SMDP and then simplify it using a structural result regarding the optimal policy. The time complexity for solving the simplified Bellman optimality equation using dynamic programming algorithms is substantially lower than that for solving the original Bellman optimality equation.

\subsection{Optimal Solution to \eqref{Problem 3}} \label{SMDP_3}

In the SMDP \eqref{Problem 3}, the system state at each decision time $A_i$ is given by the tuple $\big(\Delta(A_i), d(A_i), C_i\big)$, where $\Delta(A_i)$ denotes the AoI at time slot $A_i$, $d(A_i)$ indicates the length of the most recently received packet by that time, and $C_i$ represents the delay state in the $i$-th packet transmission epoch. A realization of the system state is denoted by $\big(\delta, d, c_i\big)$. The actions taken at each decision time $A_i$ include the waiting time $\tau_{i+1}$, the packet length $l_{i+1}$, and the buffer position $b_{i+1}$. A detailed description of this SMDP can be found in Appendix \ref{App_Comp_SMDP_3}. The Bellman optimality equation for the SMDP \eqref{Problem 3} can be expressed as
\begin{align}\label{Bellman Equation for problem 3}
&h_2\big(\delta, d, c_i\big)= \min_{\substack{\tau_{i+1} \in \{0, 1, \ldots\} \\ l_{i+1} \in \{1, 2, \ldots, B\} \\ b_{i+1} \in \{0,1,\ldots,B-l_{i+1}\}}} \nonumber\\
&\mathbb E \Bigg[ \sum_{k=0}^{\tau_{i+1}+T_{i+1}(l_{i+1})-1} \bigg(\varepsilon\big(\delta+k,d\big) -\varepsilon_{\text{opt}}\bigg) \Bigg| C_i=c_i\Bigg]\nonumber\\
&+\mathbb E \Bigg[ \sum_{k=0}^{F_{i+1}-1} \bigg(\varepsilon\big(b_{i+1}+T_{i+1}(l_{i+1})+k,l_{i+1}\big) -\varepsilon_{\text{opt}}\bigg) \Bigg| C_i=c_i\Bigg]\nonumber\\
&+\mathbb E\big[h_2\big(b_{i+1}+T_{i+1}(l_{i+1})+F_{i+1}, l_{i+1}, C_{i+1}\big)\big| C_i=c_i\big],
\end{align}
where $\delta \in \mathbb Z^+$, $d \in \{1, 2, \ldots, B\}$, $c_i \in \{1, 2, \ldots, N\}$, and $h_2\big(\delta, d, c_i\big)$ is the relative value function.

By solving the Bellman optimality equation \eqref{Bellman Equation for problem 3} using dynamic programming, we can obtain a solution to \eqref{Problem 3}. The time complexity of this dynamic programming algorithm is high, as it requires the joint optimization of three variables for each state $\big(\delta, d, c_i\big)$. The following Theorem \ref{Theorem 2} provides a threshold rule for determining the waiting time action at each~$A_i$, simplifying the Bellman optimality equation.

\begin{theorem} \label{Theorem 2}
    There exists an optimal solution $\pi^* = \big(\big(b_2^*, b_3^*, \ldots\big), \big(l_2^*, l_3^*, \ldots\big), \big(\tau_2^*, \tau_3^*, \ldots\big)\big) \in \Pi$ to problem \eqref{Problem 3}. The optimal waiting time $\tau_{i+1}^*$ is determined by the index-based threshold rule
    \begin{equation} \label{Th3_Wait_Rule}
        \tau\big(\delta, d, l_{i+1}, c_i\big) = \min\{k \geq 0 : \gamma\big(\delta+k, d, l_{i+1}, c_i\big) \geq \varepsilon_{\text{opt}}\},
    \end{equation} where $\gamma\big(\cdot\big)$ is defined in \eqref{index_function_const_packet_length}. Furthermore, the optimal packet length $l_{i+1}^*$ and buffer position $b_{i+1}^*$ are determined by solving the following simplified Bellman optimality equation:
    \begin{align}\label{Simplified Bellman Equation for problem 3}
    &h_2\big(\delta, d, c_i\big)= \min_{\substack{l_{i+1} \in \{1, 2, \ldots, B\}}} \Bigg\{\nonumber\\ &\mathbb E \Bigg[ \sum_{k=0}^{\tau\big(\delta, d, l_{i+1}, c_i\big)+T_{i+1}(l_{i+1})-1} \bigg(\varepsilon\big(\delta+k,d\big)-\varepsilon_{\text{opt}}\bigg) \Bigg| C_i=c_i \Bigg] \nonumber\\
    &+\min_{b_{i+1} \in \{0, 1, \ldots, B-l_{i+1}\}} \Bigg\{\nonumber\\
    &\mathbb E \Bigg[ \sum_{k=0}^{F_{i+1}-1} \bigg(\varepsilon\big(b_{i+1}+T_{i+1}(l_{i+1})+k,l_{i+1}\big) -\varepsilon_{\text{opt}}\bigg) \Bigg| C_i=c_i\Bigg]\nonumber\\
    &+\mathbb E\big[h_2\big(b_{i+1}+T_{i+1}(l_{i+1})+F_{i+1}, l_{i+1}, C_{i+1}\big)\big| C_i=c_i\big]\Bigg\}\Bigg\}.
    \end{align}
\end{theorem}

\begin{proof}[Proof Sketch]
    Theorem \ref{Theorem 2} is proved based on the observation that the optimal waiting time $\tau_{i+1}^*$ and buffer position $b_{i+1}^*$ for any state $\big(\delta, d, c_i\big)$ can be determined by solving separate optimization problems, provided that the optimal packet length $l_{i+1}^*$ is known. We derive a solution for the separate waiting time optimization problem as an index-based threshold rule, which, in turn, enables the simplification of the Bellman optimality equation \eqref{Bellman Equation for problem 3} to \eqref{Simplified Bellman Equation for problem 3}. The detailed proof is presented~in~Appendix \ref{App_Proof_TH2}.
\end{proof}

The Bellman optimality equations \eqref{Bellman Equation for problem 3} and \eqref{Simplified Bellman Equation for problem 3} for the SMDP \eqref{Problem 3} can be solved using policy iteration, value iteration, or linear programming \cite[Chapter 11.4.4]{puterman2014markov}. In particular, policy iteration is efficient for solving SMDPs \cite{puterman2014markov, shisher2023learning}. The time complexity for solving the simplified Bellman equation~\eqref{Simplified Bellman Equation for problem 3} via policy iteration is significantly lower than that for solving \eqref{Bellman Equation for problem 3}. Specifically, the policy improvement step for \eqref{Bellman Equation for problem 3} has a time complexity of $O(\tau_{\text{bound}}^2 B^3)$, while it is $O(\tau_{\text{bound}}^2 B + \tau_{\text{bound}} B^3 + \tau_{\text{bound}} B^2)=O(\max\{\tau_{\text{bound}}^2 B, \tau_{\text{bound}} B^3\})$ for the simplified equation \eqref{Simplified Bellman Equation for problem 3}. The policy evaluation step is identical for both equations.

\section{Simulations}
To evaluate the performance of the optimal scheduling policies in Theorems \ref{Theorem 1} and \ref{Theorem 2}, we consider two experiments: (i) remote inference of auto-regressive (AR) processes, which offers a model-based evaluation; and (ii) cart-pole state prediction, which provides a trace-driven evaluation.

\subsection{Model-Based Evaluation}

\begin{figure}[t]
\centerline{\includegraphics[width=.50\textwidth]{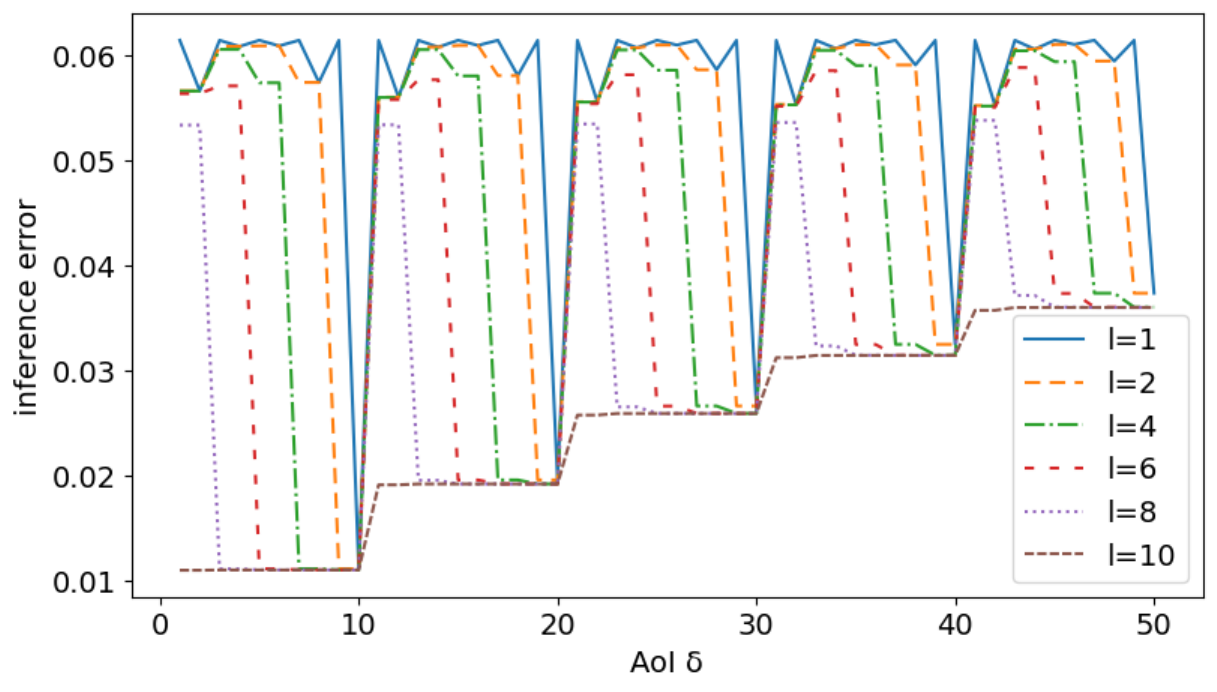}}
\caption{Inference error of the AR process for AoI values ${\delta = 1, 2, \ldots, 50}$ and packet lengths $l=1, 2, 4, 6, 8, 10$.}
\label{AR_inferr}
\end{figure}

Higher-order auto-regressive (AR) linear time-invariant systems are widely used models for various practical scenarios, such as control systems incorporating memory and delay~\cite{pham2016static,briat2010memory,lechappe2016dynamic}, and wireless communication channels~\cite{truong2013effects}. Motivated by this, we consider a remote inference problem for an AR process, where the target signal $Y_t \in \mathbb R$ evolves according to:
\begin{align}
Y_{t}=a_1 Y_{t-1}+ a_2 Y_{t-2}+\ldots+a_p Y_{t-p}+W_t,
\end{align}
where the noise $W_t \in \mathbb R$ is zero-mean Gaussian with variance~$\sigma_W^2$.  Let $V_t=Y_t+N_t$ be the noisy observation of $Y_t$, where $N_t \in \mathbb R$ is zero-mean Gaussian with variance $\sigma_N^2$.

In this experiment, we set $\sigma_W^2=0.01$, $\sigma_N^2=0.001$ and construct an AR($10$) process with coefficients $a_{2}=0.05$ and $a_{10}=0.9$. The remaining coefficients of the AR($10$) process are zero. We consider a quadratic loss function ${L_{2}(y,\hat{y})=\lVert y-\hat{y} \rVert_{2}^{2}}$. The goal is to infer the target signal $Y_t$ by using the data packet $X_{t-\delta}^l=(V_{t-\delta}, V_{t-\delta-1}, \ldots, V_{t-\delta-l+1})$. A linear MMSE estimator achieves the optimal inference error since $X_{t-\delta}^l$ and~$Y_t$ are jointly Gaussian, and the loss function is quadratic. Fig.~\ref{AR_inferr} shows the inference error for AoI values ${\delta=1, 2, \ldots, 50}$ and packet lengths $l=1, 2, 4, 6, 8, 10$.

Using this AR process, we illustrate the performance of the optimal policy derived for the scheduling problem specified in~\eqref{Problem 1: Inner opt problem}-\eqref{Problem 1: Outer opt problem} under time-invariant packet length selection. For this experiment, we set the maximum allowable packet length to~$l = 10$ and the buffer size to $B = 75$.

Let the number of network delay states be $N = 2$, implying that the delay state $C_i$ in the $i$-th epoch follows a two-state Markov chain with transition probabilities $\{p_{jk}\}_{j,k \in \{1,2\}}$. During the simulations, we assume $p_{12} = p_{21}$, ensuring that each delay state has a steady-state probability of $\frac{1}{2}$. 

We assume a simple transmission delay model that scales with the packet length and the current network conditions. The transmission delay $T_i(l)$ for each delay state is given by
\begin{align}\label{Transmit_Delay}
T_i(l) =
\begin{cases}
\lceil \sigma l \rceil, & \text{if } c_i = 1, \\
\lceil 5 \sigma l \rceil, & \text{if } c_i = 2,
\end{cases}
\end{align}
where $\sigma$ is a scaling parameter. On the other hand, the corresponding feedback delay $F_i$ is given by
\begin{align}\label{Feedback_Delay}
F_i =
\begin{cases}
1, & \text{if } c_i = 1, \\
3, & \text{if } c_i = 2.
\end{cases}
\end{align}

We control the delay memory by varying the sum of transition probabilities $\alpha = p_{12} + p_{21}$. When $\alpha = 1$, i.e., ${p_{12} = p_{21} = \frac{1}{2}}$, both delay states are equally likely in the $i$-th epoch, independent of the delay state in the $(i-1)$-th epoch. Consequently, the delay distribution is i.i.d.. However, as $\alpha$ deviates from $1$, the delay memory increases, and the delay state in the $(i-1)$-th epoch becomes more informative about the delay state in the $i$-th epoch.

Fig.~\ref{AR_perf} presents the time-average inference error achieved by the following three scheduling policies as the parameter $\sigma$ changes from $0$ to $2$ with $\alpha = \frac{1}{20}$:
\begin{itemize}
    \item[1.] Optimal policy derived for the scheduling problem \eqref{Problem 1: Inner opt problem}-\eqref{Problem 1: Outer opt problem} under time-invariant packet length selection.
    \item[2.] Policy outlined in Theorem \ref{Theorem 1} with $l=1$.
    \item[3.] Generate-at-will, zero-wait, $l=1$ policy: $(f, \ell, g)$ such that $f = g = (0, 0, \ldots)$ and $\ell = (1, 1, \ldots)$.
\end{itemize}

The result presented in Fig.~\ref{AR_perf} highlights that the optimal policy achieves a time-average inference error up to six-times lower than that of the generate-at-will, zero-wait, $l = 1$ policy. Part of this improvement can be attained using a packet length of $l = 1$ with the waiting times and buffer positions determined as described in Theorem \ref{Theorem 1}. However, a substantial portion of the performance gain comes from selecting the appropriate packet length, accounting for the interplay between the transmission delay statistics and the packet length $l$.

\begin{figure}[t]
\centerline{\includegraphics[width=.50\textwidth]{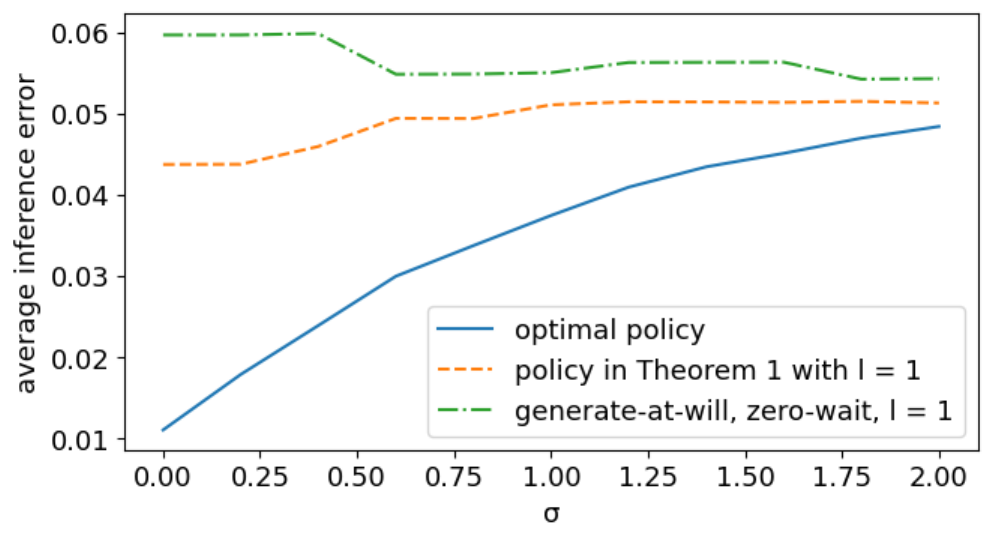}}
\caption{Time-average inference error vs. the scaling parameter~$\sigma$ of transmission delay $T_i(l)$.}
\label{AR_perf}
\end{figure}

Fig.~\eqref{AR_memory} presents the time-average inference error achieved by the following two scheduling policies as the parameter $\alpha$ varies from $0$ to $2$ with $\sigma = \frac{5}{2}$:
\begin{itemize}
    \item[1.] Policy outlined in Theorem \ref{Theorem 1} with $l=5$. This policy accounts for the delay memory.
    \item[2.] Policy outlined in \cite[Theorem 1]{shisher2023learning} with $l=5$. This policy assumes that both delay states are equally likely at each epoch, independent of the history, for any $\alpha$.
\end{itemize}

The result in Fig.~\eqref{AR_memory} emphasizes the importance of accounting for delay memory when designing scheduling algorithms. The performance of both policies is quite similar when the delay state of the network is almost i.i.d., i.e., when $\alpha$ is around~$1$. However, as $\alpha$ deviates from this value, the delay memory increases, and the advantage of the policy in Theorem~\ref{Theorem 1} becomes apparent, with a significant performance gain of up to $11.6\%$.

\begin{figure}[t]
\centerline{\includegraphics[width=.50\textwidth]{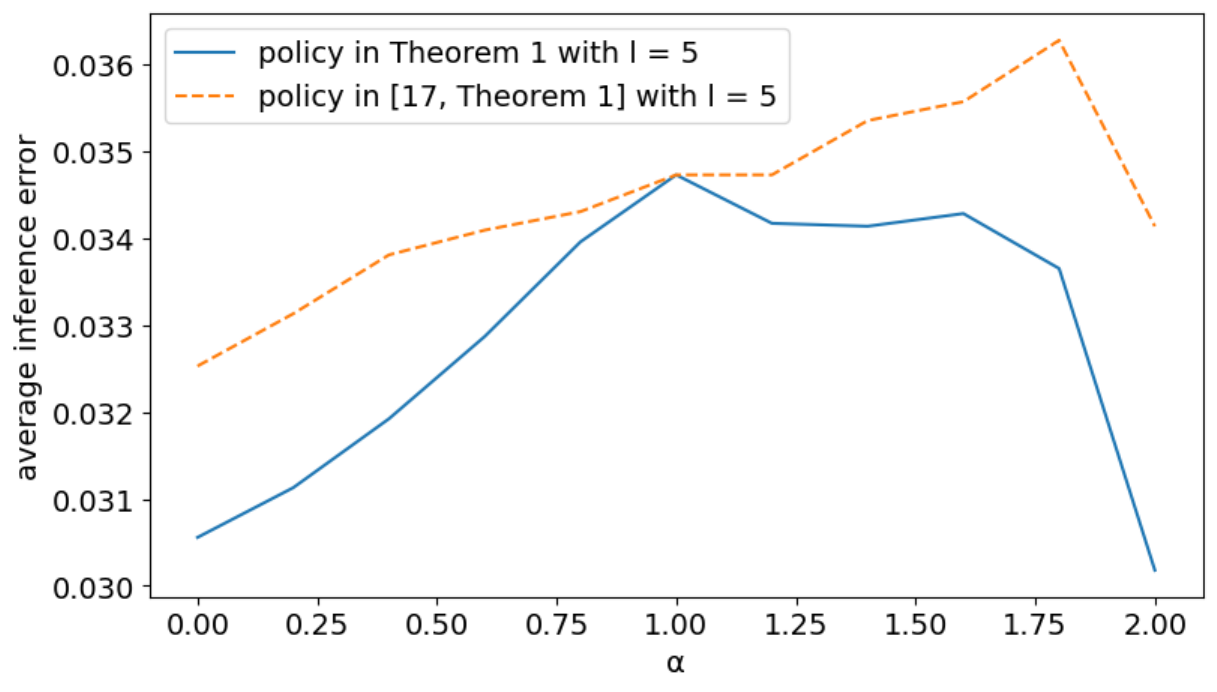}}
\caption{Time-average inference error vs. the sum of transition probabilities $\alpha = p_{12} + p_{21}$.}
\label{AR_memory}
\end{figure}

\subsection{Trace-Driven Evaluation}

\begin{figure}[t]
\centerline{\includegraphics[width=.50\textwidth]{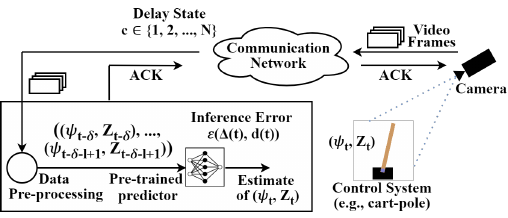}}
\caption{Structure of the cart-pole experiment.}
\label{CP_system_model}
\end{figure}

Ensuring the safe operation of robotic systems often requires evaluating their state based on observations from remote sources. This makes the prediction problem critical for monitoring the safety of such systems. We consider a specific scenario, illustrated in Fig.~\ref{CP_system_model}, where a camera remotely monitors a control system, e.g., a cart-pole. The camera captures video frames and transmits them to a receiver over a network with distinct delay states. Due to the high dimensionality of the video data, transmitting the sequence of frames takes multiple time slots. As a result, the most recently delivered sequence of video frames is generated $\delta$ time slots ago. Upon reception, the frames undergo data pre-processing to extract the cart position and the pole angle information, ${X_{t-\delta} = \big((\psi_{t-\delta}, Z_{t-\delta}), \ldots, (\psi_{t-\delta-l+1}, Z_{t-\delta-l+1})\big)}$. Subsequently, a predictor uses the most recent available information to estimate the current pole angle and cart position~${Y_t = (\psi_t, Z_t)}$.

\begin{figure}[t]
\centerline{\includegraphics[width=.50\textwidth]{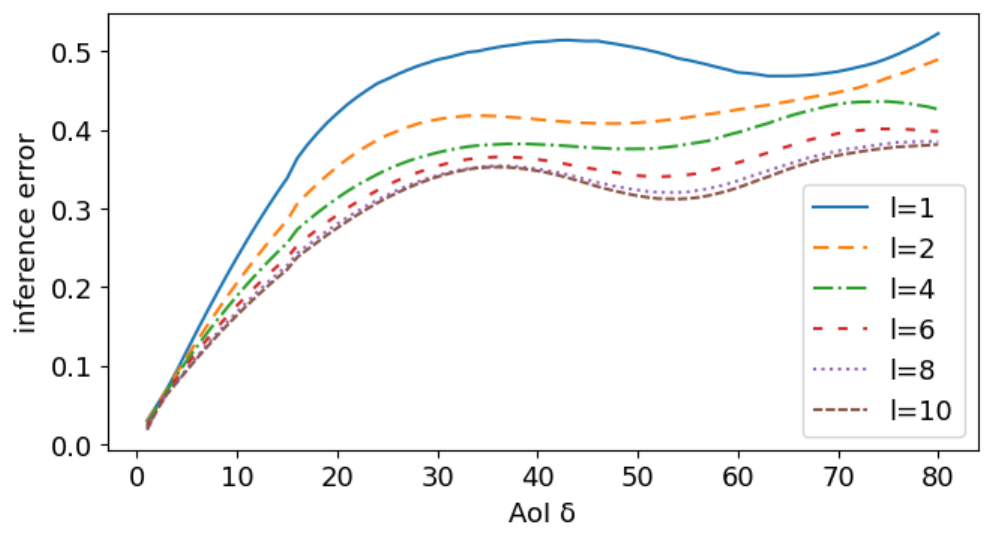}}
\caption{Inference error of the cart-pole experiment for AoI values $\delta = 1, 2, \ldots, 80$ and packet lengths ${l =1, 2, 4, 6, 8, 10}$.}
\label{CP_inferr}
\end{figure}

We consider the OpenAI CartPole-v1 task~\cite{brockman2016openai}, where a Deep Q-Network (DQN) reinforcement learning algorithm~\cite{mnih2015human} is used to control the force applied to a cart to prevent the attached pole from falling over. A pre-trained DQN neural network from \cite{shisher2024timely} is employed for this control task. A time-series dataset containing the pole angle~$\psi_t$ and cart position~$Z_t$ is obtained from \cite{shisher2024timely}. This dataset was generated by simulating 10,000 episodes of the OpenAI CartPole-v1 environment using the DQN controller.

The predictor in this experiment is a Long Short-Term Memory (LSTM) neural network consisting of one input layer, one hidden layer with 64 LSTM cells, and a fully connected output layer. The first 72\% of the dataset is used for training, and the remaining data is used for inference. The performance of the trained LSTM predictor is evaluated on the inference~dataset. The inference error is shown in Fig.~\ref{CP_inferr} for AoI values $\delta = 1, 2, \ldots, 80$ and packet lengths $l = 1, 2, 4, 6, 8, 10$. We observe that the inference error is not a monotonic function of AoI $\delta$ for a given packet length~$l$. Moreover, the inference error decreases as the packet length $l$ increases for a fixed AoI value~$\delta$.

Using this cart-pole experiment, we illustrate the performance of the optimal policy in Theorem \ref{Theorem 2}, which is derived for the scheduling problem \eqref{Problem 3} under time-variable packet length selection.

We consider a slightly more diverse network environment in this experiment by assuming three network delay states, i.e., $N = 3$. Accordingly, the delay state $C_i$ in the $i$-th epoch follows a three-state Markov chain with transition probabilities $\{p_{jk}\}_{j,k \in \{1,2,3\}}$. We set the self-transition probabilities to $p_{11} = p_{22} = p_{33} = 1 - \alpha$, and all remaining transition probabilities to $\frac{\alpha}{2}$. Due to the symmetry of the Markov chain, the steady-state probability of each delay state is $\frac{1}{3}$. Throughout the simulation, we set $\alpha = \frac{1}{50}$.

The transmission delay corresponding to each state has two components: (i) a constant $10$ time slot component, and (ii) a component that scales depending on the packet length and the network delay condition. For example, the first component may represent the unavoidable, relatively large propagation delay in non-terrestrial networks. The other part may represent the varying levels of store-and-forward operation of the satellites, as well as the increase in the transmission delay due to packet lengths. As a result, the transmission delay $T_i(l)$ for each delay state is given by
\begin{align}\label{Transmit_Delay_2}
T_i(l) = 10+
\begin{cases}
\lceil \sigma l \rceil, & \text{if } c_i = 1, \\
\lceil 2\sigma l \rceil, & \text{if } c_i = 2, \\
\lceil 3\sigma l \rceil, & \text{if } c_i = 3,
\end{cases}
\end{align} where $\sigma$ is a scaling parameter. On the other hand, the feedback delay $F_i$ for each delay state is given by
\begin{align}\label{Transmit_Delay_2}
F_i =
\begin{cases}
1, & \text{if } c_i = 1, \\
2, & \text{if } c_i = 2, \\
3, & \text{if } c_i = 3.
\end{cases}
\end{align}

\begin{figure}[t]
\begin{minipage}{.24\textwidth}
    \subfloat[$\sigma=\frac{1}{5}$]{\includegraphics[width=\textwidth]{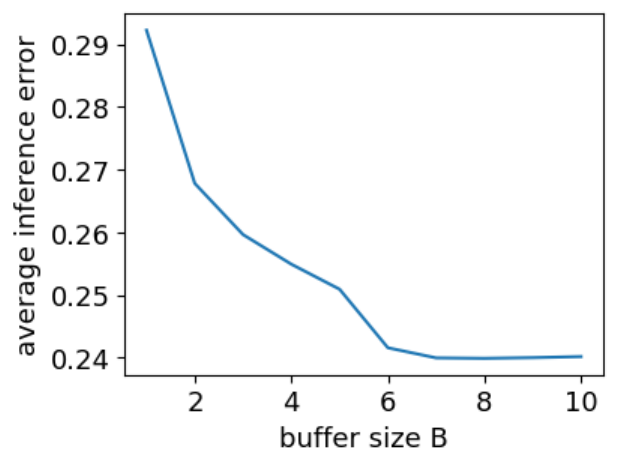}\label{CP_BS_02}}
\end{minipage}
\hfill    
\begin{minipage}{.24\textwidth}
    \subfloat[$\sigma=\frac{9}{5}$]{\includegraphics[width=\textwidth]{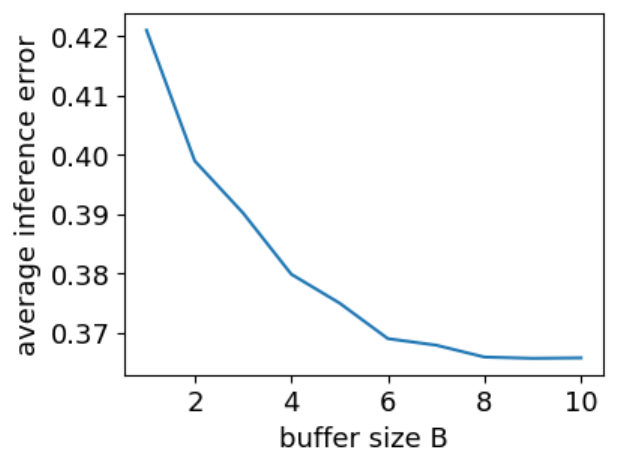}\label{CP_BS_18}}
\end{minipage}
    \caption{Time-average inference error vs. the buffer size B for ${\sigma=\frac{1}{5}}$ and ${\sigma=\frac{9}{5}}$.}\label{CP_BS}
\end{figure}

\begin{figure}[t]
\centerline{\includegraphics[width=.50\textwidth]{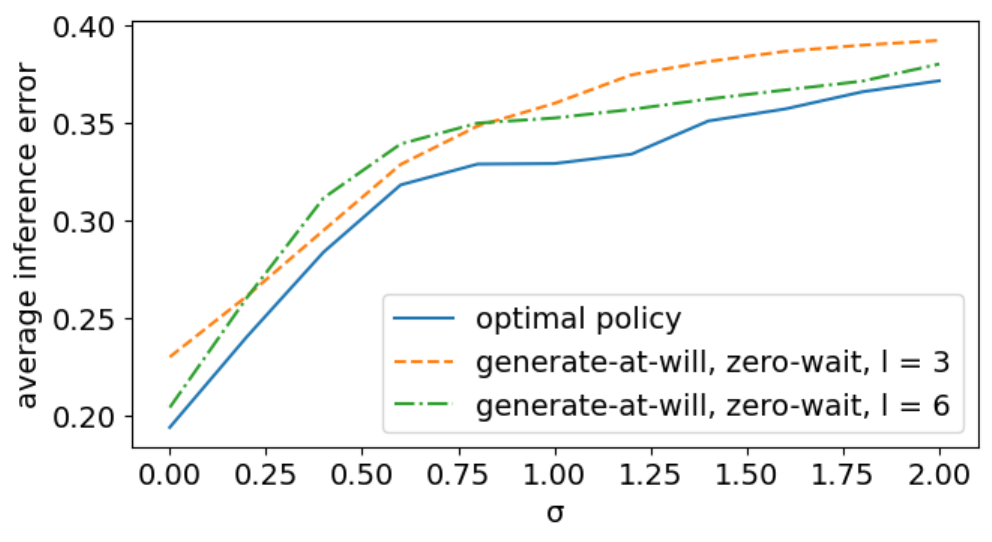}}
\caption{Time-average inference error vs. the scaling parameter $\sigma$ of transmission delay $T_i(l)$.}
\label{CP_perf}
\end{figure}

Fig.~\ref{CP_BS}(a) and Fig.~\ref{CP_BS}(b) show the time-average inference error achieved by the optimal policy in Theorem \ref{Theorem 2} for buffer size values $B$, ranging from $1$ to $10$, with ${\sigma=\frac{1}{5}}$ and ${\sigma=\frac{9}{5}}$, respectively. As the buffer size increases, the performance improvement diminishes and becomes less significant. Therefore, we fix the buffer size $B$ at $10$ for this experiment.

Fig.~\ref{CP_perf} presents the time-average inference error achieved by the following three scheduling policies as the parameter $\sigma$ changes from $0$ to $2$:
\begin{itemize}
    \item[1.] Optimal policy in Theorem \ref{Theorem 2}
    \item[2.] Generate-at-will, zero-wait, $l = 3$ policy: $(f, \ell, g)$ such that $f = g = (0, 0, \ldots)$ and $\ell = (3, 3, \ldots)$.
    \item[3.] Generate-at-will, zero-wait, $l = 6$ policy: $(f, \ell, g)$ such that $f = g = (0, 0, \ldots)$ and $\ell = (6, 6, \ldots)$.
\end{itemize}

The result in Fig.~\ref{CP_perf} highlights the importance of adjusting the packet length over time in response to varying delay conditions. The optimal policy in Theorem \ref{Theorem 2} outperforms the other two policies, which use constant packet length, achieving up to $15.7\%$ reduction in average inference error.

\begin{table}[t]
\centering
    \caption{Normalized time complexity comparison for solving the original Bellman optimality equation (Orig.~BE)~\eqref{Bellman Equation for problem 3} and the simplified Bellman optimality equation (Simp.~BE)~\eqref{Simplified Bellman Equation for problem 3}, both with and without the pre-computed $\gamma(\cdot)$ in~\eqref{index_function_const_packet_length}, using the policy iteration algorithm, for buffer sizes $B \in \{1,5,10\}$.}
\label{tab:complexity_comparison}
\footnotesize
\setlength{\tabcolsep}{6pt}
\renewcommand{\arraystretch}{1.2}
\begin{tabular}{lccc}
\hline
 & \multicolumn{3}{c}{\textbf{Buffer size $B$}} \\
\cline{2-4}
\textbf{Bellman equation} & $B=1$ & $B=5$ & $B=10$ \\
\hline
Orig.~BE 
& $\mathcal{T}_{1}$ & $\mathcal{T}_{5}$ & $\mathcal{T}_{10}$ \\
\hline
Simp.~BE 
& $\approx \mathcal{T}_{1}$ & $0.8556\mathcal{T}_{5}$ & $0.4527\mathcal{T}_{10}$ \\
\hline
Simp.~BE + pre-computed $\gamma(\cdot)$ 
& $\approx \mathcal{T}_{1}$ & $0.8531\mathcal{T}_{5}$ & $0.4464\mathcal{T}_{10}$ \\
\hline
\end{tabular}
\end{table}

We finally present an empirical comparison of the normalized time complexity for solving the original Bellman optimality equation~\eqref{Bellman Equation for problem 3} and the simplified Bellman optimality equation~\eqref{Simplified Bellman Equation for problem 3}, provided in Section~\ref{sec_4}, using the policy iteration algorithm. For the simplified Bellman optimality equation~\eqref{Simplified Bellman Equation for problem 3}, the time complexity is reported for two cases: when the index function $\gamma(\cdot)$ in~\eqref{index_function_const_packet_length} is pre-computed and when it is not initially available. The results are summarized in Table~\ref{tab:complexity_comparison} for buffer sizes $B \in \{1,5,10\}$. Let $\mathcal{T}_k$ denote the time required to solve the original Bellman equation for buffer size~$B=k$. The corresponding solution times for the simplified Bellman equation, both with and without pre-computed~$\gamma(\cdot)$, are expressed relative to $\mathcal{T}_k$.

For a small buffer size of $B = 1$, the time complexities of solving the two equations are nearly identical. This is because, when $B = 1$, the action space is extremely limited: the only possible choice for the buffer position action $b_{i+1}$ at decision time $A_i$ is zero, and the only possible packet length action $l_{i+1}$ is one. Consequently, the reduction in complexity of the policy improvement step of the policy iteration algorithm is not observable. In contrast, as the buffer size increases, solving the simplified Bellman equation has notably lower time complexity than solving the original Bellman optimality equation, achieving more than a two-fold reduction at buffer size $B = 10$. Moreover, the computation of the index function $\gamma(\cdot)$, which is required for solving the simplified Bellman equation, incurs only a small computational overhead.

\section{Conclusion}
This paper studies a goal-oriented communication design problem for remote inference, where an intelligent neural network model on the receiver side predicts the real-time value of a target signal using data packets transmitted from a remote location. We derive two optimal scheduling policies under time-invariant and time-variable packet length selection. These policies minimize the expected time-average inference error for a possibly non-monotonic inference error function by considering the interplay between packet length and transmission delay, as well as by exploiting delay memory. Finally, through the simulations, we demonstrate that our goal-oriented scheduler drops inference error down to one-sixth with respect to age-based scheduling of unit-length packets.

\section*{Acknowledgment}
The authors would like to thank Batu Saatci for his assistance with the simulations.

\bibliographystyle{IEEEtran}
\bibliography{JS}

\begin{IEEEbiography}[{\includegraphics[width=1in,height=1.25in,clip,keepaspectratio]{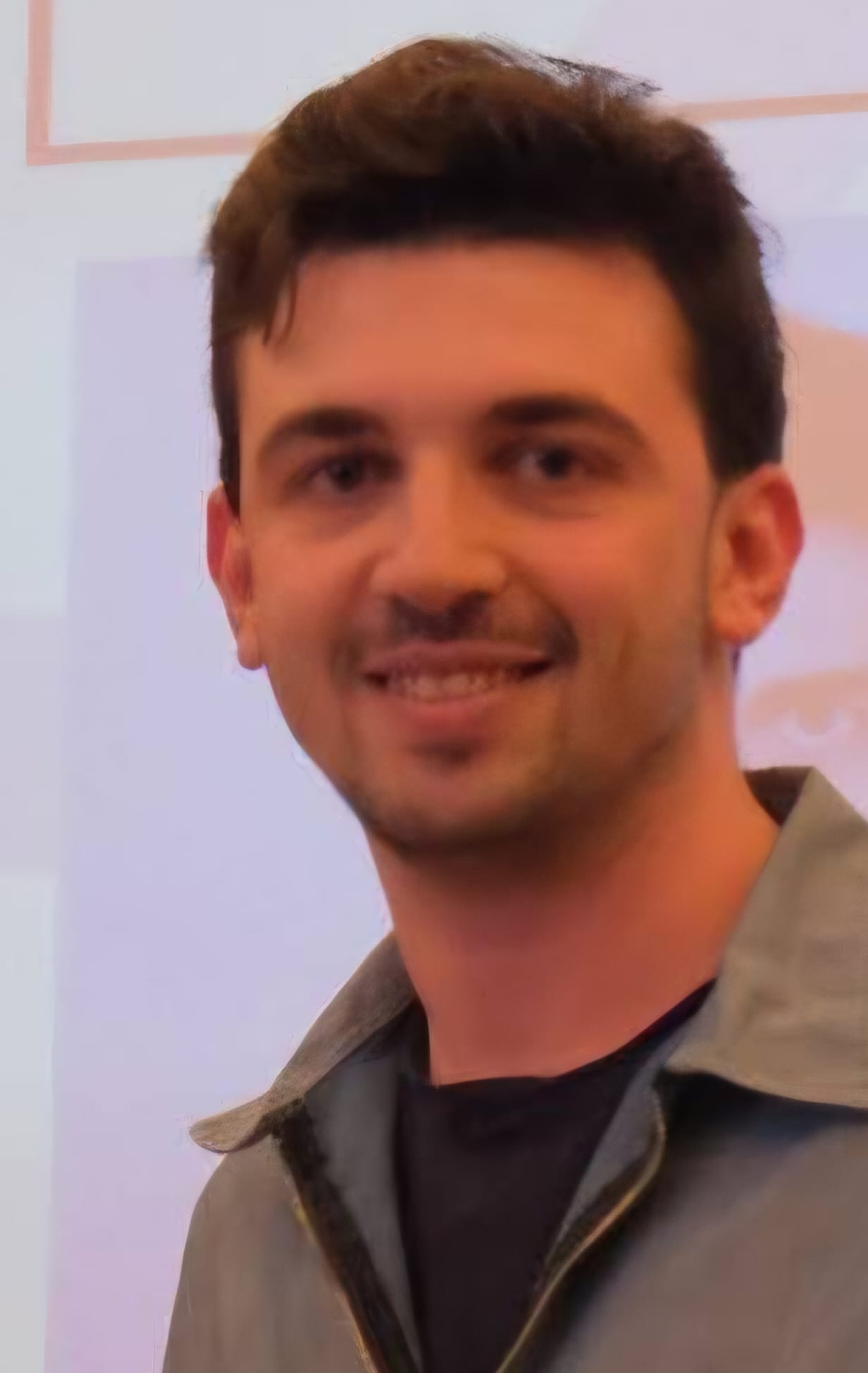}}]{Çağrı Arı} received the B.S. degree and the M.S. degree in Electrical and Electronics Engineering from Middle East Technical University (METU), Ankara, Turkiye, in 2023 and 2025, respectively. He is currently pursuing a PhD degree in the same department.  His research interests include Goal-oriented Communications, Random Access, and Age of Information. He has served as a reviewer for IEEE TIT and IEEE TCOM.
\end{IEEEbiography}

\begin{IEEEbiography}[{\includegraphics[width=1in,height=1.25in,clip,keepaspectratio]{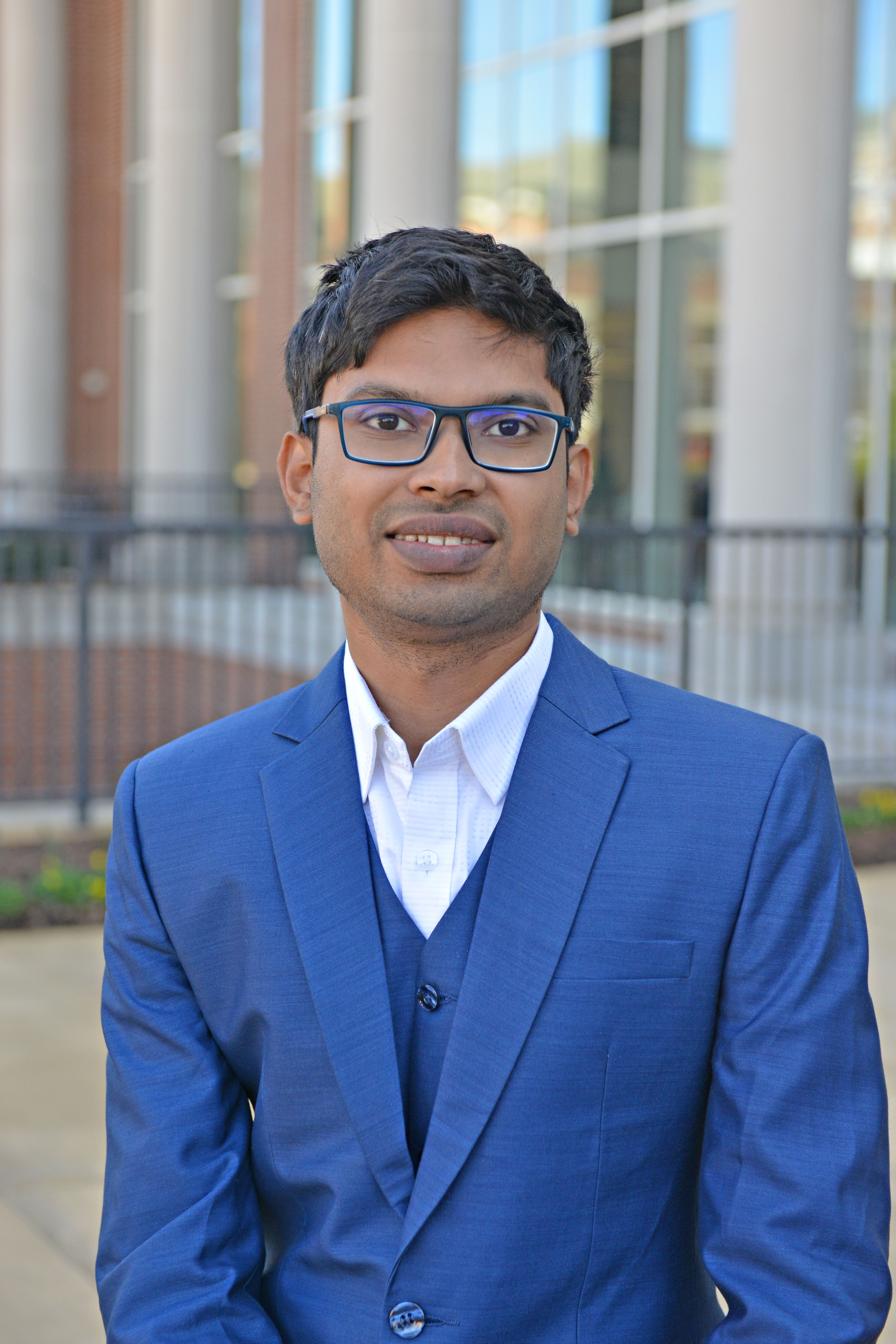}}]{Md Kamran Chowdhury Shisher}(Member, IEEE) received Ph.D. in Electrical Engineering from Auburn University, AL, USA, in 2024. He is currently a Postdoctoral Researcher with the Department of Electrical and Computer Engineering, Purdue University, West Lafayette, IN, USA. His research focuses on the stochastic optimization and sequential decision-making problems in Communication Networks, Machine Learning, and Multi-agent Networked Control systems. He has received the IEEE Communications Society William R. Bennett Prize in 2025.
\end{IEEEbiography}

\begin{IEEEbiography}[{\includegraphics[width=1in,height=1.25in,clip,keepaspectratio]{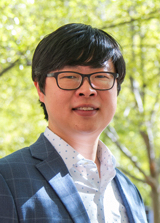}}]{Yin Sun} (Senior Member, IEEE) is the Bryghte D. and Patricia M. Godbold Endowed Associate Professor in the Department of Electrical and Computer Engineering at Auburn University, Alabama. He received his B.Eng. and Ph.D. degrees in Electronic Engineering from Tsinghua University in 2006 and 2011, respectively.
From 2011 to 2017, he was a Postdoctoral Scholar and Research Associate at The Ohio State University. He joined Auburn University in 2017 as an Assistant Professor and was promoted to Associate Professor in 2023. His research interests include Semantic and Goal-oriented Communications, Wireless Networks, and Applied Artificial Intelligence in Agriculture. Dr. Sun has served on the editorial boards of the \emph{IEEE/ACM Transactions on Networking}, \emph{IEEE Transactions on Information Theory}, \emph{IEEE Transactions on Network Science and Engineering}, \emph{IEEE Transactions on Green Communications and Networking}, and the \emph{Journal of Communications and Networks}. He has also served on the organizing committees of numerous international conferences, including as Technical Program Committee Chair for ACM MobiHoc 2025 and General Chair for IEEE/IFIP WiOpt 2026. He founded the Age of Information (AoI) Workshop in 2018 and the Modeling and Optimization in Semantic Communications (MOSC) Workshop in 2023. His publications have received multiple recognitions, including the Best Student Paper Award at IEEE/IFIP WiOpt 2013, the Best Paper Award at IEEE/IFIP WiOpt 2019, runner-up for the Best Paper Award at ACM MobiHoc 2020, the Best Paper Award from the \emph{Journal of Communications and Networks} in 2021, and the IEEE Communications Society William R. Bennett Prize in 2025. He received the Auburn Author Award in 2020 and the National Science Foundation (NSF) CAREER Award in 2023. 
\end{IEEEbiography}

\begin{IEEEbiography}[{\includegraphics[width=1in,height=1.25in,clip,keepaspectratio]{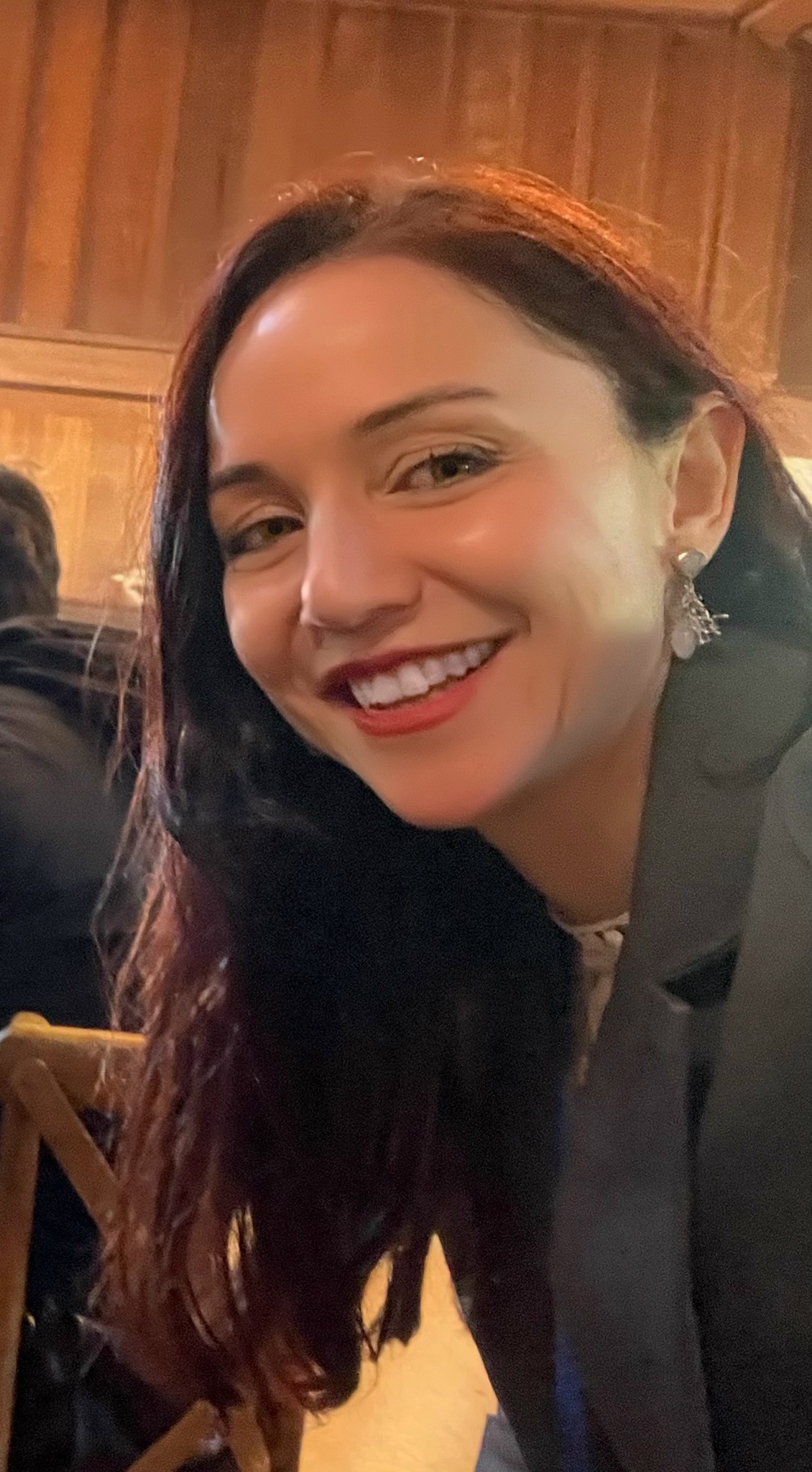}}]{Elif Uysal} (Fellow, IEEE) is a Professor of Electrical and Electronics Engineering at the Middle East Technical University (METU), in Ankara, Turkey. She received the Ph.D. degree from Stanford University in 2003, the S.M. degree from the Massachusetts Institute of Technology (MIT) in 1999, and the B.S. degree in 1997 from METU. She was elected IEEE Fellow for “pioneering contributions to energy-efficient and low-latency communications” and Fellow of the Artificial Intelligence Industry Alliance in 2022. She was a co-recipient of the IEEE INFOCOM2026 Test of Time Paper Award and the 2025 Selcuk Yasar Award. Dr. Uysal received an ERC Advanced Grant in 2024 for her project GO SPACE (Goal Oriented Networking for Space), a TUBITAK BIDEB National Pioneer Researcher Grant (2020), and the Science Academy of Turkey Young Scientist Award (2014). She has chaired the METU Parlar Foundation for Education and Research since 2022. She is the founder of FRESHDATA Technologies (2022). She has served as associate editor for IEEE Trans. on Wireless Communications, IEEE Trans. on Networking, IEEE Trans. on Information Theory, and Guest Editor at IEEE J. Special Areas in Information Theory. 
\end{IEEEbiography}

\appendices
\section{} \label{App_Comp_SMDP_1}
The SMDP corresponding to problem \eqref{Problem 1: Inner opt problem}, constructed in Section \ref{SMDP_1}, is described by the following components \cite{bertsekasdynamic}:

\textbf{(i) Decision Time:} Each ACK reception time slot $A_i = S_i + T_i(l) + F_i$ is a decision time.

\textbf{(ii) Action:} At each decision time $A_i$, the scheduler determines the waiting time $\tau_{i+1}$ and buffer position $b_{i+1}$ for the next packet to be transmitted.

\textbf{(iii) State and State Transitions:} The system state at each decision time $A_i$ is represented by the tuple $\big(\delta, c_i\big)$, where $\delta = \Delta(A_i)$ and $c_i$ is the delay state in the $i$-th epoch. The AoI $\Delta(t)$ at time slot $t$ evolves according to equation \eqref{AoI Evolution}. Meanwhile, the delay state $C_i = c_i$ in the $i$-th epoch evolves according to the finite-state ergodic Markov chain defined in Section \ref{sys_model}. The Markov chain makes a single transition at each decision time $A_i$ and none otherwise.

\textbf{(iv) Transition Time and Cost:} The time between the consecutive decision times $A_{i+1} - A_i$ is represented by the random variable $\tau_{i+1} + T_{i+1}(l) + F_{i+1}$, given the delay state $C_i = c_i$. Moreover, the cost incurred while transitioning from decision time $A_i$ to decision time $A_{i+1}$,
$$\sum_{t=A_i}^{A_{i+1}-1}\varepsilon\big(\Delta(t),l\big),$$ is represented by the random variable
$$\sum_{k=0}^{\tau_{i+1}+T_{i+1}(l)-1} \varepsilon\big(\Delta(A_i)+k,l\big)+\sum_{k=0}^{F_{i+1}-1} \varepsilon\big(b_{i+1}+T_{i+1}(l)+k,l\big),$$ given the delay state $C_i = c_i$.

\section{Proof of Theorem \ref{Theorem 1}} \label{App_Proof_TH1}
This proof consists of three steps.

\textbf{Step 1:} We prove two structural results regarding the optimal scheduling policy by using the Bellman optimality equation \eqref{Bellman Equation for fixed packet length}, which comprises three terms. The first term,
\begin{align*}
\mathbb{E} \Bigg[ \sum_{k=0}^{\tau_{i+1} + T_{i+1}(l) - 1} \bigg(\varepsilon\big(\delta + k, l\big) - \varepsilon_{l, \text{opt}}\bigg) \Bigg| C_i = c_i \Bigg],
\end{align*}
is independent of the buffer position $b_{i+1}$. In contrast, the remaining two terms,
\begin{align*}
\mathbb{E} \Bigg[ \sum_{k=0}^{F_{i+1} - 1} \bigg(\varepsilon\big(b_{i+1} + T_{i+1}(l) + k, l\big) - \varepsilon_{l, \text{opt}}\bigg) \Bigg| C_i = c_i \Bigg]
\end{align*}
and
\begin{align*}
\mathbb{E} \bigg[ h_1\big(b_{i+1} + T_{i+1}(l) + F_{i+1}, C_{i+1}\big) \bigg| C_i = c_i \bigg],
\end{align*}
are independent of the waiting time $\tau_{i+1}$. Consequently, the optimal buffer position $b_{i+1}^{*}$ and the optimal waiting time $\tau_{i+1}^*$ can be determined by solving separate optimization problems.

\textbf{(i) Optimal Waiting Time:} The optimal waiting time $\tau_{i+1}^*$, when $\Delta(A_i)=\delta$ and $C_i=c_i$, can be determined by solving the optimization problem
\begin{align}\label{Optimal waiting time}
\min_{\tau_{i+1} \in \{0, 1, \ldots\}}\mathbb E \Bigg[ \sum_{k=0}^{\tau_{i+1}+T_{i+1}(l)-1} \bigg(\varepsilon\big(\delta+k,l\big) -\varepsilon_{l, \text{opt}}\bigg) \Bigg| C_i=c_i\Bigg].
\end{align}

Problem \eqref{Optimal waiting time} is a simple integer optimization problem. The optimal waiting time $\tau_{i+1}^*=0$ if 
\begin{align}\label{Opt_W_time_S1}
\min_{\tau_{i+1} \in \{1, 2, \ldots\}}&\mathbb E \Bigg[ \sum_{k=0}^{\tau_{i+1}+T_{i+1}(l)-1} \bigg(\varepsilon\big(\delta+k,l\big) -\varepsilon_{l, \text{opt}}\bigg) \Bigg| C_i=c_i\Bigg] \geq \nonumber\\ &\mathbb E \Bigg[ \sum_{k=0}^{T_{i+1}(l)-1} \bigg(\varepsilon\big(\delta+k,l\big) -\varepsilon_{l, \text{opt}}\bigg) \Bigg| C_i=c_i\Bigg].
\end{align}

By using the linearity of expectation, we can obtain
\begin{align}\label{Opt_W_time_S2}
&\min_{\tau_{i+1} \in \{1, 2, \ldots\}} \mathbb E \Bigg[ \sum_{k=0}^{\tau_{i+1}-1} \bigg(\varepsilon\big(\delta+T_{i+1}(l)+k,l\big) -\varepsilon_{l, \text{opt}}\bigg) \Bigg| C_i=c_i\Bigg] \nonumber \\
&\geq 0.
\end{align}

The inequality \eqref{Opt_W_time_S2} holds if and only if
\begin{align}\label{Opt_W_time_S3}
\min_{\tau_{i+1} \in \{1, 2, \ldots\}} \frac{1}{\tau_{i+1}}  \sum_{k=0}^{\tau_{i+1}-1} \mathbb{E} \bigg[\varepsilon\big(\delta+T_{i+1}(l)+k,l\big)\bigg| C_i=c_i \bigg] \geq \varepsilon_{l, \text{opt}}.
\end{align}

The left-hand side of \eqref{Opt_W_time_S3} is equal to \(\gamma\big(\delta, l, l, c_i\big)\), where $\gamma\big(\cdot\big)$ is the index function defined in \eqref{index_function_const_packet_length}. Similarly, the optimal waiting time $\tau_{i+1}^*=1$, if $\tau_{i+1}^*\neq 0$ and 
\begin{align}\label{Opt_W_time_S4}
&\min_{\tau_{i+1} \in \{2, 3, \ldots\}}\mathbb E \Bigg[ \sum_{k=1}^{\tau_{i+1}-1} \bigg(\varepsilon\big(\delta+T_{i+1}(l)+k,l\big) -\varepsilon_{l, \text{opt}}\bigg) \Bigg| C_{i}=c_i\Bigg] \nonumber \\
&\geq 0.
\end{align}

Following the same steps in \eqref{Opt_W_time_S2}-\eqref{Opt_W_time_S3}, we can demonstrate that the inequality \eqref{Opt_W_time_S4} is equivalent to
\begin{align}\label{Opt_W_time_S5}
\gamma\big(\delta+1, l, l, c_i\big)\geq \varepsilon_{l, \text{opt}}.
\end{align}

By repeating \eqref{Opt_W_time_S4}-\eqref{Opt_W_time_S5}, we can establish the rule that the optimal waiting time $\tau_{i+1}^*$ is equal to $k$, provided that \({\tau_{i+1}^* \neq 0, 1, \ldots, k-1}\) and 
\begin{align}
\gamma\big(\delta+k, l, l, c_i\big)\geq \varepsilon_{l, \text{opt}}.
\end{align}

Hence, any \(\tau_{i+1}^*\) in the optimal waiting time sequence ${g^* = \big(\tau_2^*, \tau_3^*, \ldots\big)}$ is determined by the index-based threshold rule \(\tau_l\big(\delta, c_i\big)\) given by equation \eqref{Theorem_1_Wait_Rule}.

\textbf{(ii) Optimal Buffer Position:} The optimal buffer position $b_{i+1}^*$, when $\Delta(A_i) = \delta$ and $C_i = c_i$, can be determined by solving the optimization problem
\begin{align}\label{opt_buff_pos}
&\min_{b_{i+1} \in \{0,1,\ldots,B-l\}} \nonumber\\
&\mathbb{E} \Bigg[ \sum_{k=0}^{F_{i+1} - 1} \bigg(\varepsilon\big(b_{i+1} + T_{i+1}(l) + k, l\big) - \varepsilon_{l, \text{opt}}\bigg) \Bigg| C_i = c_i \Bigg]\nonumber\\
&+\mathbb E\bigg[h_1\big(b_{i+1}+T_{i+1}(l)+F_{i+1}, C_{i+1}\big)\bigg| C_i=c_i\bigg].
\end{align}

From \eqref{opt_buff_pos}, we observe that the optimal buffer position $b_{i+1}^*$ is independent of $\delta = \Delta(A_i)$ and depends only on the delay state $c_i$ in the previous epoch. In other words, $c_i$ is a sufficient statistic for determining $b_{i+1}^*$ for all $i$.

\textbf{Step 2:} We construct a new SMDP based on the results obtained in the first step. In the new SMDP, the decision time is $D_{i+1}$ instead of $A_i$. At each decision time $D_{i+1}$, the system state is represented only by the delay state $c_i$ in the $i$-th epoch. We fix the waiting time decisions by using the optimal threshold rule $\tau_l\big(\delta, c_i\big)$, and the action involves determining the buffer position $b_{i+1}$. A detailed description of this SMDP can be found in Appendix \ref{App_Comp_SMDP_2}. The Bellman optimality equation of this new SMDP is given by equation \eqref{New Bellman Equation for fixed packet length}.

The second term on the right-hand side of equation \eqref{New Bellman Equation for fixed packet length}, $$\mathbb E\bigg[h_1'\big(C_{i+1}\big)\bigg| C_i=c\bigg],$$ is independent of the action $b_{i+1}$ taken at decision time $D_{i+1}$ and depends only on the state $C_i=c_i$. Therefore, the Bellman optimality equation \eqref{New Bellman Equation for fixed packet length} is decomposable and can be solved as a per-decision epoch optimization problem; that is, we can express any $b_{i+1}^*$ in the optimal buffer position sequence ${f^* = \big(b_2^*, b_3^*, \ldots\big)}$ as shown in \eqref{Theorem_1_Buf_Pos_Seq}.

\textbf{Step 3:} We determine the optimal value $\varepsilon_{l, \text{opt}}$ of the inner optimization problem \eqref{Problem 1: Inner opt problem} as the unique root of equation \eqref{Theorem_1_unique_root_of}.  Let $\vartheta_i^*$ denote the random variable $b_{i}^*+T_{i}(l)+F_{i}$. Consider the following equality:
\begin{align} \label{unique_root_step_1}
    &\mathbb{E}\Bigg[\sum_{t=A_{i}}^{A_{i+1}-1} \varepsilon\big(\Delta(t)\big)\Bigg]-\varepsilon_{l, \text{opt}}\mathbb{E}\big[A_{i+1}-A_{i}\big]=\nonumber\\
    &\mathbb E \Bigg[ \sum_{k=0}^{\tau_l\big(\vartheta_i^*, C_i\big)+T_{i+1}(l)-1} \bigg(\varepsilon\big(\vartheta_i^*+k,l\big)-\varepsilon_{l, \text{opt}} \bigg)\Bigg]\nonumber\\
    &+\mathbb E \Bigg[ \sum_{k=0}^{F_{i+1}-1} \bigg(\varepsilon\big(b_{i+1}^*+T_{i+1}(l)+k,l\big)-\varepsilon_{l, \text{opt}} \bigg)\Bigg].
\end{align}

We first demonstrate that the right-hand side of equation \eqref{unique_root_step_1} is equal to zero, which implies \(\varepsilon_{l, \text{opt}}\) is indeed a root of equation \eqref{Theorem_1_unique_root_of}. From the Bellman equation \eqref{Bellman Equation for fixed packet length}, we can obtain the following relation:
\begin{align}\label{unique_root_step_2}
&h_1\big(\vartheta_i^*, C_i\big)= \nonumber\\
&\mathbb E \Bigg[ \sum_{k=0}^{\tau_l\big(\vartheta_i^*, C_i\big)+T_{i+1}(l)-1} \bigg(\varepsilon\big(\vartheta_i^*+k,l\big) -\varepsilon_{l, \text{opt}}\bigg) \Bigg| C_i\Bigg]\nonumber\\
&+\mathbb E \Bigg[ \sum_{k=0}^{F_{i+1}-1} \bigg(\varepsilon\big(b_{i+1}^*+T_{i+1}(l)+k,l\big) -\varepsilon_{l, \text{opt}}\bigg) \Bigg| C_i\Bigg]\nonumber\\
&+\mathbb E\bigg[h_1\big(\vartheta_{i+1}^*, C_{i+1}\big)\bigg| C_i\bigg].
\end{align}

By the law of iterated expectations, taking the expectation over $C_i$ on both sides of \eqref{unique_root_step_2} gives us the equation
\begin{align}\label{unique_root_step_3}
&\mathbb E \big[ h_1\big(\vartheta_i^*, C_i\big)\big]= \nonumber\\
&\mathbb E \Bigg[ \sum_{k=0}^{\tau_l\big(\vartheta_i^*, C_i\big)+T_{i+1}(l)-1} \bigg(\varepsilon\big(\vartheta_i^*+k,l\big) -\varepsilon_{l, \text{opt}}\bigg)\Bigg]\nonumber\\
&+\mathbb E \Bigg[ \sum_{k=0}^{F_{i+1}-1} \bigg(\varepsilon\big(b_{i+1}^*+T_{i+1}(l)+k,l\big) -\varepsilon_{l, \text{opt}}\bigg)\Bigg]\nonumber\\
&+\mathbb E\big[h_1\big(\vartheta_{i+1}^*, C_{i+1}\big)\big].
\end{align}

Since $C_i$ evolves according to an ergodic Markov chain with a unique stationary distribution, it follows that $\mathbb{E} \big[ h_1\big(\vartheta_i^*, C_i\big) \big] = \mathbb{E}\big[h_1\big(\vartheta_{i+1}^*, C_{i+1}\big)\big]$. Therefore, equation \eqref{unique_root_step_3} implies that the right-hand side of equation \eqref{unique_root_step_1} is equal to zero.

Next, we prove the uniqueness of the root \(\varepsilon_{l, \text{opt}}\). The right-hand side of equation \eqref{unique_root_step_1} can be written as
\begin{align} \label{unique_root_step_4}
    &\min_{\tau\big(\cdot\big)} \text{ } \mathbb E \Bigg[ \sum_{k=0}^{\tau\big(\vartheta_i^*, C_i\big)+T_{i+1}(l)-1} \bigg(\varepsilon\big(\vartheta_i^*+k,l\big)-\varepsilon_{l, \text{opt}} \bigg)\Bigg]\nonumber\\
    &+\mathbb E \Bigg[ \sum_{k=0}^{F_{i+1}-1} \bigg(\varepsilon\big(b_{i+1}^*+T_{i+1}(l)+k,l\big)-\varepsilon_{l, \text{opt}} \bigg)\Bigg].
\end{align}

The expression in \eqref{unique_root_step_4} is concave, continuous, and strictly decreasing in \(\varepsilon_{l, \text{opt}}\), as the first term is the functional minimum of linear decreasing functions of \(\varepsilon_{l, \text{opt}}\), and the second term is itself a linear decreasing function of \(\varepsilon_{l, \text{opt}}\). Therefore, since
$$
\lim_{\varepsilon_{l, \text{opt}} \to \infty} \mathbb{E} \Bigg[ \sum_{k=0}^{\tau\big(\vartheta_i^*, C_i\big)+T_{i+1}(l)-1} \bigg(\varepsilon\big(\vartheta_i^*+k, l\big) - \varepsilon_{l, \text{opt}} \bigg) \Bigg] = -\infty
$$
and
$$
\lim_{\varepsilon_{l, \text{opt}} \to -\infty} \mathbb{E} \Bigg[ \sum_{k=0}^{\tau\big(\vartheta_i^*, C_i\big)+T_{i+1}(l)-1} \bigg(\varepsilon\big(\vartheta_i^*+k, l\big) - \varepsilon_{l, \text{opt}} \bigg) \Bigg] = \infty
$$
for any $\tau\big(\cdot\big)$, and
$$
\lim_{\varepsilon_{l, \text{opt}} \to \infty} \mathbb{E} \Bigg[ \sum_{k=0}^{F_{i+1}-1} \bigg(\varepsilon\big(b_{i+1}^* + T_{i+1}(l) + k, l\big) - \varepsilon_{l, \text{opt}} \bigg) \Bigg] = -\infty
$$
and
$$
\lim_{\varepsilon_{l, \text{opt}} \to -\infty} \mathbb{E} \Bigg[ \sum_{k=0}^{F_{i+1}-1} \bigg(\varepsilon\big(b_{i+1}^* + T_{i+1}(l) + k, l\big) - \varepsilon_{l, \text{opt}} \bigg) \Bigg] = \infty,
$$
equation \eqref{Theorem_1_unique_root_of} has a unique root. This completes the proof.

\section{} \label{App_Comp_SMDP_3}
The SMDP corresponding to problem \eqref{Problem 3}, constructed in Section \ref{SMDP_3}, is described by the following components \cite{bertsekasdynamic}:

\textbf{(i) Decision Time:} Each ACK reception time slot $A_i = S_i + T_i(l_i) + F_i$ is a decision time.

\textbf{(ii) Action:} At each decision time $A_i$, the scheduler determines the waiting time $\tau_{i+1}$, packet length $l_{i+1}$, and buffer position $b_{i+1}$ for the next packet to be transmitted.

\textbf{(iii) State and State Transitions:} The system state at each decision time $A_i$ is represented by the tuple $\big(\delta, d, c_i\big)$, where $\delta = \Delta(A_i)$, $d = d(A_i)$, and $c_i$ is the delay state in the $i$-th epoch. The AoI $\Delta(t)$ at time slot $t$ and the length $d(t)$ of the packet most recently delivered by time slot $t$ evolve according to equations \eqref{AoI Evolution} and \eqref{length of the most recently delivered packet}, respectively. Meanwhile, the delay state $C_i = c_i$ in the $i$-th epoch evolves according to the finite-state ergodic Markov chain defined in Section \ref{sys_model}. The Markov chain makes a single transition at each decision time $A_i$ and none otherwise.

\textbf{(iv) Transition Time and Cost:} The time between the consecutive decision times $A_{i+1} - A_i$ is represented by the random variable $\tau_{i+1} + T_{i+1}(l_{i+1}) + F_{i+1}$, given the delay state $C_i = c_i$. Moreover, the cost incurred while transitioning from decision time $A_i$ to decision time $A_{i+1}$,
$$\sum_{t=A_i}^{A_{i+1}-1}\varepsilon\big(\Delta(t),d(t)\big),$$ is represented by the random variable
\begin{align}
    & \sum_{k=0}^{\tau_{i+1}+T_{i+1}(l_{i+1})-1} \varepsilon\big(\Delta(A_i)+k,d(A_i)\big) \nonumber \\
    & +\sum_{k=0}^{F_{i+1}-1} \varepsilon\big(b_{i+1}+T_{i+1}(l_{i+1})+k,l_{i+1}\big), \nonumber
\end{align} given the delay state $C_i = c_i$.

\section{} \label{App_Comp_SMDP_2}
The components of the SMDP constructed in the second step of the proof of Theorem \eqref{Theorem 1} differ from those of the SMDPs constructed in Sections \ref{SMDP_1} and \ref{SMDP_3}.

\textbf{(i) Decision Time:} Each delivery time slot $D_{i+1} = S_{i+1} + T_{i+1}(l)$ is a decision time.

\textbf{(ii) Action:} At each decision time $D_{i+1}$, the scheduler determines the buffer position $b_{i+1}$ for the next packet to be transmitted.

\textbf{(iii) State and State Transitions:} The system state at each decision time $D_{i+1}$ is the delay state $C_i = c_i$ in the $i$-th epoch, which evolves according to the finite-state ergodic Markov chain defined in Section \ref{sys_model}. The Markov chain makes a single transition at each decision time $A_i$ and none otherwise.

\textbf{(iv) Transition Time and Cost:} 
The time between consecutive decision times, $D_{i+1} - D_i$, is represented by the random variable $\vartheta_{i+1} = F_{i+1} + \tau_l\big(b_{i+1} + T_{i+1}(l) + F_{i+1}, C_{i+1}\big) + T_{i+2}(l)$, given the delay state $C_i = c_i$, where $\tau_l\big(\cdot\big)$ is the waiting time rule defined in \eqref{Theorem_1_Wait_Rule}. Moreover, the cost incurred while transitioning from decision time $D_i$ to decision time $D_{i+1}$,
$$\sum_{t=D_i}^{D_{i+1}-1} \varepsilon\big(\Delta(t), l\big),$$
is represented by the random variable
$$\sum_{k=0}^{\vartheta_{i+1}-1} \varepsilon\big(b_{i+1} + T_{i+1}(l) + k, l\big),$$
given the delay state $C_i = c_i$.

\section{Proof of Theorem \ref{Theorem 2}} \label{App_Proof_TH2}
The Bellman optimality equation in \eqref{Bellman Equation for problem 3} has three terms. The first term,
\begin{align*}
\mathbb E \Bigg[ \sum_{k=0}^{\tau_{i+1}+T_{i+1}(l_{i+1})-1} \bigg(\varepsilon\big(\delta+k,d\big) -\varepsilon_{\text{opt}}\bigg) \Bigg| C_i=c_i\Bigg],
\end{align*}
does not depend on the buffer position $b_{i+1}$. On the contrary, the remaining two terms,
\begin{align*}
\mathbb E \Bigg[ \sum_{k=0}^{F_{i+1}-1} \bigg(\varepsilon\big(b_{i+1}+T_{i+1}(l_{i+1})+k,l_{i+1}\big) -\varepsilon_{\text{opt}}\bigg) \Bigg| C_i=c_i\Bigg]
\end{align*}
and
\begin{align*}
\mathbb E\big[h_2\big(b_{i+1}+T_{i+1}(l_{i+1})+F_{i+1}, l_{i+1}, C_{i+1}\big)\big| C_i=c_i\big],
\end{align*}
do not depend on the waiting time $\tau_{i+1}$. Therefore, given the optimal packet length \(l_{i+1}^*\), the optimal waiting time \(\tau_{i+1}^*\), when \(\Delta(A_i) = \delta\), \(d(A_i) = d\), and \(C_i = c_i\), can be determined by solving the following optimization problem, which is independent of the action \(b_{i+1}\):
\begin{align}\label{Optimal waiting time Problem 3}
\min_{\tau_{i+1} \in \{0, 1, \ldots\}}\mathbb E \Bigg[ \sum_{k=0}^{\tau_{i+1}+T_{i+1}(l_{i+1}^*)-1} \bigg(\varepsilon\big(\delta+k,d\big) -\varepsilon_{\text{opt}}\bigg) \Bigg| C_i=c_i\Bigg].
\end{align}

Problem \eqref{Optimal waiting time Problem 3} is a simple integer optimization problem, quite similar to problem \eqref{Optimal waiting time}. By following the steps in \eqref{Opt_W_time_S1} through \eqref{Opt_W_time_S5}, we can show that the optimal waiting time $\tau_{i+1}^*=k$, if $\tau_{i+1}^*\neq 0, 1, \ldots, k-1$ and 
\begin{align}
\gamma(\delta+k, d, l_{i+1}^*, c_i)\geq \varepsilon_{\text{opt}},
\end{align}  where $\gamma\big(\cdot\big)$ is the index function defined in \eqref{index_function_const_packet_length}. Hence, the optimal waiting time \(\tau_{i+1}^*\) is determined as given in equation \eqref{Th3_Wait_Rule}, and the Bellman optimality equation \eqref{Bellman Equation for problem 3} can be equivalently formulated as \eqref{Simplified Bellman Equation for problem 3}. This completes the proof.

\end{document}